\documentclass[USenglish,oneside]{article}
\pdfoutput=1
\usepackage[utf8]{inputenc}
\usepackage{arxiv_style}
\usepackage{amssymb}
\usepackage{amsmath}
\usepackage{mathtools}
\usepackage{amsthm}
\usepackage{xspace}
\usepackage{multicol}
\usepackage{xcolor}
\usepackage{url}
\usepackage{algorithmic, algorithm}
\usepackage[T1]{fontenc}
\usepackage{graphicx}
\usepackage{bmpsize}

\let\svthefootnote\thefootnote
\newcommand\blankfootnote[1]{%
  \let\thefootnote\relax\footnotetext{#1}%
  \let\thefootnote\svthefootnote%
}

\newcounter{ag}
\addtocounter{ag}{1}

\newcounter{ab}
\addtocounter{ab}{1}

\newcounter{ar}
\addtocounter{ar}{1}

\newcounter{irg}
\addtocounter{irg}{1}

\newcounter{igh}
\addtocounter{igh}{1}

\newcounter{ms}
\addtocounter{ms}{1}

\newtheorem{theorem}{Theorem}
\newtheorem{corollary}{Corollary}
\newtheorem{definition}{Definition}

\newcommand{\crit}{\ensuremath{t^*}\xspace}
\newcommand{\dbsize}{\ensuremath{N}\xspace}
\newcommand{\eps}{\ensuremath{\varepsilon}\xspace}
\newcommand{\epsfrac}{\ensuremath{\rho}\xspace}

\newcommand{\grand}{\bar{y}}
\renewcommand{\k}{\ensuremath{k}\xspace}
\newcommand{\lap}{\ensuremath{{\sf Lap}}\xspace}

\newcommand{\mse}{\textit{MSE}\xspace}
\newcommand{\msa}{\textit{MSA}\xspace}

\newcommand{\normal}{\ensuremath{\mathcal{N}}}

\newcommand{\sse}{\textit{SSE}\xspace}
\newcommand{\ssa}{\textit{SSA}\xspace}

\newcommand{\se}{\textit{SE}\xspace}
\newcommand{\sa}{\textit{SA}\xspace}
\newcommand{\sqa}{\textit{SQA}\xspace}
\newcommand{\sqe}{\textit{SQE}\xspace}

\newcommand{\spa}{\textit{SPA}\xspace}

\newcommand{\var}{\textit{VAR}\xspace}
\newcommand{\varq}{\ensuremath{\var_q}\xspace}
\newcommand{\x}{\ensuremath{\mathbf{x}}\xspace}
\newcommand{\xprime}{\ensuremath{\mathbf{x'}}\xspace}
\newcommand{\X}{\ensuremath{\mathbf{X}}\xspace}
\newcommand{\data}{\ensuremath{\mathbf{x}}\xspace}
\newcommand{\Data}{\ensuremath{\mathbf{X}}\xspace}

\newcommand{\E}{\text{E}}

\newcommand{\yj}{\bar{y}_j}
\newcommand{\yb}{\bar y}

\newtheorem{prop}{Proposition}

\title{ Improved Differentially Private Analysis of Variance}
\author{
\textbf{Marika Swanberg}\\
Mathematics Department\\
Reed College, Portland, OR\\
\texttt{marlswanb@reed.edu}
\And
Ira Globus-Harris\\
Mathematics Department\\
Reed College, Portland, OR\\
 \texttt{irglobush@reed.edu}
 \And
 Iris Griffith\\
Mathematics Department\\
Reed College, Portland, OR\\
 \texttt{irisrose@reed.edu}
 \And
 Anna Ritz\\
Biology Department\\
Reed College, Portland, OR\\
 \texttt{aritz@reed.edu}
 \And
 Adam Groce\thanks{Corresponding authors.}\\
Mathematics Department\\
Reed College, Portland, OR\\
 \texttt{agroce@reed.edu} 
 \And
 Andrew Bray$^*$\\
Mathematics Department\\
Reed College, Portland, OR\\
\texttt{abray@reed.edu}
}

\date{}
\begin{document}

\maketitle

\begin{abstract}
{Hypothesis testing is one of the most common types of data analysis and forms the backbone of scientific research in many disciplines.  Analysis of variance (ANOVA) in particular is used to detect dependence between a categorical and a numerical variable.  Here we show how one can carry out this hypothesis test under the restrictions of differential privacy.  We show that the $F$-statistic, the optimal test statistic in the public setting, is no longer optimal in the private setting, and we develop a new test statistic $F_1$ with much higher statistical power.  We show how to rigorously compute a reference distribution for the $F_1$ statistic and give an algorithm that outputs accurate $p$-values.  We implement our test and experimentally optimize several parameters.  We then compare our test to the only previous work on private ANOVA testing, using the same effect size as that work.  We see an order of magnitude improvement, with our test requiring only 7\% as much data to detect the effect.}
\end{abstract}

\section{Introduction}

A universal and recurring challenge in scientific research is determining whether a measured effect is real.  That is, researchers wish to determine if the effect observed in a particular dataset indicates a similar effect in the broader world from which the sample was drawn. The most common statistical tool to make this determination is a hypothesis test.
The particular form of the hypothesis test is driven by the scientific question and the data at hand. 

Hypothesis testing is a common tool in population association studies, where the goal is to identify whether genetic variation is associated with disease risk~\cite{balding2006tutorial}.  Consider a study looking at the effect of a particular gene's mutation on some health outcome (e.g., blood pressure or weight).  The data may include the mutation status of that gene (which may harbor one or more mutations on one or both copies of DNA).  
The researcher's goal is to determine if the gene's mutation status has an impact on that health outcome.  When the health outcome is measured using a numerical variable, the natural hypothesis test to use is the one-way ANOVA (analysis of variance) test, treating the gene's mutation status as a categorical variable.

The first step in conducting the ANOVA is to calculate a single number, the $F$-statistic, which measures the variation in group means compared to the variation in individual data points. The $F$-statistic is constructed so that if the expected value of the health outcome is the same in all groups, the expected value of $F$ is 1.  If this is not the case, the value can be dramatically larger.  Seeing a high value of $F$, a researcher will compare this value to the \textit{reference distribution}, i.e., the distribution of $F$ that would occur if the gene had no impact on the health outcome.  The result of this comparison is a $p$-value, the probability that the observed $F$ could occur by chance.  If the $p$-value is low, the analyst can conclude that this gene must indeed affect the given health outcome.  (For more detail on how ANOVA is used in this setting, see~\cite{myers2003researchdesign}.)

The analysis described above assumes that the researcher has full access to the database. However, there are many settings in medicine, psychology, education, and economics (not to mention private-sector data analysis) where the database is not available to the analyst due to privacy concerns. A well-established solution is to allow the researcher to issue queries to the data which are proven to satisfy differential privacy.  Differential privacy requires the addition of random noise to statistical queries and guarantees that the results reveal very little about any individual's data.

In this paper we propose a new statistic for ANOVA, called $F_1$, that is specifically tailored to the differentially private setting. This statistic measures the same variations as the $F$ statistic, but uses $|a-b|$ instead of $(a-b)^2$ to measure the distance between $a$ and $b$.  In the public setting the $F_1$ is a worse test statistic than the traditional $F$-statistic, but we show that in the private setting it has much higher power than the previously published differentially private $F$-statistic.  That is, we show that it can detect effects with a little as 7\% of the data that was previously required.  (In one example, an effect that took 5300 data points to detect 90\% of the time with $\epsilon=1$ in the prior work takes only 350 data points to detect using our new hypothesis test.)


\subsection{Contributions and organization}

We first review differential privacy, hypothesis testing, and the body of work that lies at the intersection of the two fields (Section~\ref{sec:background}).  
In Section \ref{sec:f1} we then present a new test statistic, $F_1$, for ANOVA in the private setting.  While there is some work on differentially private hypothesis testing, designing a new test statistic explicitly tailored for compatibility with differential privacy has been done by few others~\cite{rogers2017new}.  

In Section 3.2 we give a private algorithm for computing an approximation of $F_1$ by applying the Laplacian mechanism to the computation of several intermediate values.  Section 3.3 then describes how to compute the correct reference distribution for $F_1$ to in order to compute accurate $p$-values, which are the end result used by practitioners.  Computing the reference distribution is complicated by the fact that, unlike the traditional $F$-statistic, $F_1$ is not scale-free. 

We implement the private $F_1$-statistic and apply the method to different simulated datasets in Section~\ref{sec:results}.  The computational experiments allow us to optimize $\rho$, a parameter that determines the allocation of our privacy budget between the two important intermediate values.  We also compare our method to prior work~\cite{campbell2018diffprivanova}, and show an order of magnitude improvement in statistical power.


Finally, in Section~\ref{sec:considerations} we present a generalization of the $F_1$ statistic that allows for an arbitrary exponent in the distance measure (besides the absolute value from $F_1$ and the $L^2$-norm from the traditional $F$). We find that the $L^1$-norm used in the $F_1$ statistic is the most powerful across a wide range of scenarios.

\section{Background}
\label{sec:background}
We begin by discussing hypothesis testing in general, the one-way analysis of variance (ANOVA) test in particular, and differential privacy.  Readers familiar with one or more of these topics should be able to skip the relevant sections.  We then discuss how these topics come together in private hypothesis testing and related work in this area.

\subsection{Hypothesis Testing}
Hypothesis tests are common tools for making statistical inferences from data. The end goal of a hypothesis test is to determine whether a data set is consistent with a proposed model. This model is called the \textit{null hypothesis}, denoted $H_0$, and it suggests a mechanism by which the data could have been generated. The mechanism is chosen to be scientifically meaningful, for example: the variable of interest has the same distribution across all of the treatment groups.

The comparison between $H_0$ and an observed data set is made using a \textit{test statistic}.  A test statistic $f$ is simply a function from the data set to the real numbers.  The goal is to design a test statistic with a known distribution when the data comes from $H_0$, but which will follow a markedly different distribution under other scenarios.  The question then becomes, for a given database $\data$ with $f(\data)=t$, how likely is a value at least as extreme as $t$ to occur if $\data$ was drawn from $H_0$.  To compute this probability, we need to compare $t$ to the reference distribution.

\begin{definition}[Reference Distribution] \label{def:refdist}
Suppose $f$ is a function that computes a test statistic. The reference distribution for $f$ is the probability distribution of the statistic $T$ when $T=f(\Data)$ and $\Data$ is drawn from a distribution consistent with $H_0$.
\end{definition}

This reference distribution is used to calculate a $p$-value. A $p$-value is the probability, under the reference distribution, of drawing a statistic at least as extreme as the observed statistic.

\begin{definition}[$p$-value] \label{def:pvalue}
For a given test statistic $t=f(\x)$ and null hypothesis $H_0$, the $p$-value is defined as
\begin{equation*}
\Pr[T\geq t \mid T = f(\X) \text{ and } \X \leftarrow H_0].
\end{equation*}
\end{definition}

The $p$-value provides context for the observed statistic by positioning it in the range of statistics that could be observed under $H_0$.

Typically, researchers choose a significance threshold $\alpha$ and reject the null hypothesis when their calculated $p$-value is less than $\alpha$. The $\alpha$-level determines the probability of a \textit{type I error}, which occurs when an analyst rejects the null hypothesis despite it being true. The value of the statistic that demarcates this rejection region is called the \emph{critical value}, denoted by \crit. That is, $\Pr[T\geq \crit \mid T = f(\X) \text{ and } \X \leftarrow H_0] = \alpha$.

When one develops a test statistic, a primary goal is to maximize \textit{statistical power}. The power of a test quantifies how effectively it can detect a deviation from $H_0$. It is the probability of rejecting when $H_0$ is false.  Generally, the power depends on both the amount of data and the effect (i.e., how different the true distribution $H_A$ is from $H_0$).

\begin{definition}[Statistical Power] \label{def:power}
For a specific alternate hypothesis $H_A$, the statistical power of a hypothesis test is 
\begin{equation*}
\Pr[T \geq \crit \mid T = f(\X) \text{ and } \X \leftarrow H_A]
\end{equation*}
\end{definition}

\subsection{One-way ANOVA}

Consider again our example from earlier, where a researcher has, for a set of individuals, both blood pressure measurements and the mutation status of a particular gene.  This is a classic setting for a one-way analysis of variance (ANOVA) test.

Index each individual person or observation in a database $\x$ with $i \in \{1, \ldots, N\}$. Each $i$ is associated with a group or category $c_i$ (e.g., mutation status) and a numerical value $y_i \in \mathbb{R}$ (e.g., blood pressure).  We use $\bar{y}$ to represent the mean of all $N$ numerical values. Index each group with $j \in \{1, \ldots, k\}$. Each $j$ is associated with $C_j = \{i \mid c_i=j\}$, the set of indices of observations in group $j$. Denote the size of the set $n_j$ and the mean of the values in that set $\bar{y}_j$.  (Note that this means $\bar{y}_{c_i}$ is the mean of numerical values in the same group as observation $i$.)

The null hypothesis in a one-way ANOVA is that the $y_i$ follow identical normal distributions regardless of their group. This motivates the test statistic used in an ANOVA test, the $F$-statistic, which measures the ratio of the variation between group means (weighted by the size of the groups) and the variation of individuals within each group. If all groups had equal means, the variation between group means would be proportional to the variation between individual observations.\footnote{Our approach relies upon simulating normally distributed data in correspondence with the traditional normality assumption. The one-way ANOVA test is known to be robust to deviations from normality~\cite{schmider2010}, so our approach should be applicable even in settings where the normality assumption is suspect. Readers interested in more about ANOVA generally are referred to \cite{cox1974theoretical}.}

\begin{definition}[$F$-Statistic] \label{def:fstat}
Given a database \x with $k$ groups and \dbsize total entries, the $F$-statistic is the ratio of two values, traditionally called $\ssa(\x)$ and $\sse(\x)$.  The Sum of Squared errors of All category means ($\ssa$) is a measure of variance between group means, weighted by group size:
\begin{equation*}
\ssa(\x) = \sum_{j=1}^{\k} n_j (\bar{y}_j - \bar{y})^2.
\end{equation*}
The Sum of Squared Errors of all observations ($\sse$) is a measure of variance within groups:
\begin{equation*}
\sse(\x) = \sum_{i=1}^{\dbsize}  (y_{i}-\bar{y}_{c_i})^2.
\end{equation*}
The $F$-statistic is the ratio of \ssa and \sse, each divided by their respective degrees of freedom.  These adjusted values are called  the Mean Sum of All category errors (\msa) and Mean of Sum of Squared Errors (\mse).
We can now finish defining the $F$-statistic.
\begin{equation*}
F(\x)  = \frac{\ssa(\x)/(\k - 1)}{\sse(\x)/(\dbsize - \k)} = \frac{\msa(\x)}{\mse(\x)}.
\end{equation*}
\end{definition}

Under the null hypothesis that the $y_i$ follow identical normal distributions regardless of their group, the reference distribution of the $F$-statistic is known exactly. This comes from recognizing that $\ssa(\x)$ is drawn from $\sigma^2\chi_{\k-1}^2$, the chi-squared distribution with $\k-1$ degrees of freedom scaled by within-group variance, and $\sse(\x)$ is drawn from $\sigma^2 \chi^2_{\dbsize - \k}$.  The ratio of these values, therefore, has a reference distribution that is scale-free (not dependent on $\sigma$) and can be calculated knowing only $N$ and $k$.

\subsection{Differential Privacy}
Differential privacy is a security definition for the release of information about a database of private records. Here we outline the foundational definitions and theorems in differential privacy; everything below first appeared in the seminal paper of Dwork et al.~\cite{dwork2006calibrating}.

Suppose we have a database \x containing sensitive information that we want to study. In particular, we want to publish the output $f(\x)$ of a function $f$ (also sometimes called a \textit{mechanism}) on our database while protecting the privacy of the individuals whose data was collected. Differential privacy promises that an adversary will learn approximately nothing about an individual as a result of their presence in \x.  Informally, this is done by requiring that the probability of seeing any particular output is roughly the same regardless of what information a given individual submitted to the database.

As above, we use \dbsize to represent the number of rows in \x, where a ``row'' is simply the set of data associated with a single individual.  We now define \textit{neighboring databases}, which differ in only one row.

\begin{definition}[Neighboring Databases]\label{def:Neighboring} Two databases \x and $\x'$ are \textit{neighboring} if \x can be transformed to $\x'$ by changing only one individual's data (where a change is an in-place modification, not a full addition or removal).
\end{definition}

To protect the privacy of individuals in a database, differentially private requires that the output of a query on any two neighboring databases should look nearly identical.

\begin{definition}[Differential Privacy] \label{def:diffpriv}
A (randomized) mechanism $f$ with range R is \emph{$\eps$-differentially private} if for all $S \subseteq R$ and for all neighboring databases \x and $\x'$
\begin{equation*}
\text{\emph{Pr}}[f(\x) \in S] \leq e^{\eps} \text{\emph{Pr}}[f(\x') \in S]. 
\end{equation*}
\end{definition}

The parameter \eps is called the \textit{privacy parameter}, and its choice is a policy decision. The lower the chosen \eps, the stronger the privacy guarantee. Note that because neighboring databases are the same size, $N$ can always be released without compromising privacy.\footnote{Differential privacy can also be defined in terms of databases that differ by an addition/deletion, rather than by a change in a row.  For most applications these definitions are equivalent except for a change in $\epsilon$ by a factor of two (with the version here being the more stringent interpretation of $\epsilon$). The one significant difference is that the other definition does not result in $N$ being public, which is important for our work here.}

Differential privacy has several useful properties.  One of the most useful is composition:

\begin{theorem}[Composition]\label{thm:composition} Suppose $f$ and $g$ are respectively $\eps_{1}$- and $\eps_{2}$-differentially private mechanisms. Then, a mechanism $h$ that returns the results of applying $f$ and $g$ to \x, $h(\x) = (f(\x), g(\x))$, is $(\eps_{1} + \eps_{2})$-differentially private.
\end{theorem}

In other words, the privacy guarantee decreases, but does not disappear, when a database is queried multiple times.  Composition allows database administrators to issue researchers a privacy budget, which researchers can then divide up as they wish between any number of different queries.

Another defining feature of differential privacy is its resistance to \textit{post-processing}.

\begin{theorem}[Post-Processing] \label{thm:postprocessing}
Let $f$ be an $\eps$-differentially private mechanism, and let $g$ be an arbitrary function. Then, $h(\x) = g(f(\x))$ is also $\eps$-differentially private.
\end{theorem}

This theorem allows us to do any computation we desire on the output of our differentially-private mechanism without diminishing the privacy guarantees. We will utilize this property to compute $p$-values of the private $F$-statistic.  The $p$-values will be automatically private without additional argument.

Our algorithms are constructed by taking building blocks and combining them with composition and post-processing, but the fundamental building blocks are made private using the Laplace mechanism, the oldest and maybe simplest method for achieving differential privacy.  The Laplace mechanism allows the conversion of any function $f$ into a private approximation $\hat{f}$. One must first compute (or bound) the \textit{sensitivity} of the function, the maximum effect on the output that a single row can have.

\begin{definition}[Sensitivity] \label{def:sensitivity}
The sensitivity of a (deterministic) real-valued function $f$ on databases is the maximum of  $\lvert f(\x) - f(\x') \rvert$ taken over all pairs $(\x, \x')$ of neighboring databases.\footnote{Sensitivity and the Laplace mechanism can be defined on functions with output in $\mathbb{R}^n$, but we only need the one-dimensional version.}
\end{definition}

The Laplace mechanism will use random noise drawn from the Laplace distribution.

\begin{definition}[Laplace Distribution] \label{def:laplacedist}
The Laplace Distribution (centered at 0) with scale $b$ is the distribution with probability density function
\begin{equation*}
\lap(z\mid b) = \frac{1}{2b}\text{\emph{exp}}\bigg({-\frac{\lvert z \rvert}{b}}\bigg).
\end{equation*}
We use $\lap(b)$ to represent the Laplace distribution with scale $b$.
\end{definition}

We can now present the Laplace mechanism.

\begin{theorem}[Laplacian Mechanism] \label{thm:lapmechanism}
Let $f$ be a function with sensitivity bounded above by $s$. Let $L$ be a random variable drawn from $\lap(s/\eps)$. Then the function $\hat{f}(\x) = f(\x) + L$ is \eps-differentially private. 
\end{theorem}

\subsection{Differentially Private Hypothesis Testing}
In order to create a differentially private hypothesis test, we need a private function $f$ of a database to serve as our test statistic. This could be a differentially private estimate of an existing test statistic, or it could be a new test statistic altogether. Because randomization is essential to differential privacy, $f$ will be randomized. The same statistic on the same database may yield different outputs each time it is computed.

In addition to a test statistic $f$, we require a suitable reference distribution to calculate the corresponding $p$-value. While it may be tempting to compute the $p$-value using the reference distribution for the non-private statistic one is estimating, this may yield wildly inaccurate results~\cite{campbell2018diffprivanova}, because adding noise to the statistic increases the probability of outlier output values. Instead, we must compute the reference distribution for the noisy statistic. Only then can we calculate an accurate $p$-value. 

The goal of differentially private hypothesis testing is to create a private test statistic and method of computing the $p$-value that maximizes statistical power, ideally approaching the power of the equivalent test in the classical non-private setting.

\subsection{Related Work}

There has been a moderate amount of work on differentially private hypothesis testing, but because there are many hypothesis tests most individual tests have received only a small amount of attention, and some very common tests have not seen a private analogue developed at all.

Several papers have addressed testing the value of a mean or the difference of means \cite{solea2014differentially, d2015differential, ding2018comparing}.  Hypothesis tests using coefficients of a linear regression to test for dependence between continuous variables is extremely common in many academic disciplines, but only recently has a method for carrying this analysis out privately been developed \cite{sheffet2015differentially, barrientos2017differentially}. and Nguy{\^e}n and Hui propose a test for surival analysis data \cite{nguyen2017differentially}.  There is one prior work on private ANOVA testing, that of  Campbell et al.~\cite{campbell2018diffprivanova}.  We will discuss this result in greater depth in the next section.

The chi-squared test, which tests for the independence of two categorical variables, has received the most study.  Vu and Slavkovi\'{c} \cite{vu2009differential} give an analogue to the test and also compute accurate $p$-values.   Many private chi-squared tests have been specifically motivated by genome-wide association studies (GWAS)  \cite{fienberg2011privacy, uhlerop2013privacy, johnson2013privacy}.  These give p-value calculations, but they are only accurate in the limit as $N$ grows large.  Other work has used Monte Carlo simulations (as we do in this work) to give more accurate reference distributions for small $N$ \cite{gaboardi2016differentially, wang2015revisiting}.  Rogers and Kifer \cite{rogers2017new} instead propose a new statistic with an asymptotic distribution more similar to its non-private analogue.  We note that this is one of few papers that, like the present work, proposes test statistics intended for the private setting, rather than simply approximating the accepted test statistic from the classic public setting. Very few of these papers carefully measure the power of the test they develop.  Rogers and Kifer \cite{rogers2017new} and Gaboardi et al.~\cite{gaboardi2016differentially} are notable exceptions, giving power curves for several different approaches.

There is also a significant body of work looking at how quickly private approximations of test statistics converge to their limiting distributions (e.g., \cite{smith2008efficient, wasserman2010statistical, smith2011privacy}).  These are important theoretical results, but they do not usually yield practical tests.  Unless $N$ is very large (in which case the details of the test do not matter very much anyway) the distribution of the test statistic is not close enough to that of the standard public to allow accurate computation of $p$-values.

\subsection{Prior work on private ANOVA}

The only previous work on differentially private ANOVA testing that the authors are aware of is Campbell et al.~\cite{campbell2018diffprivanova} Using the ANOVA test as defined above, they analyze the sensitivity of the \ssa and \sse with the assumption that all data was normalized to be between $0$ and $1$ and add Laplacian noise proportional to these sensitivities to the public computation of the \ssa and \sse. Their algorithm then uses post-processing to calculate the noisy $F$-statistic, and returns this in addition to the noisy \ssa and \sse (Algorithm~\ref{alg:Fhat}). 
\begin{algorithm}
    \begin{algorithmic}
        \STATE Compute $\widehat{\text{SSA}} = \text{SSA} + Z_1$ where $Z_1\sim\text{Lap}\left(\frac{7 - 9/N}{\eps/2}\right)$
        \STATE Compute $\widehat{\text{SSE}} = \text{SSE} + Z_2$ where $Z_2\sim\text{Lap}\left(\frac{5-4/N}{\eps/2}\right)$
        \STATE Compute $\widehat{F} = \frac{\widehat{\text{SSA}}/(\k-1)}{\widehat{\text{SSE}}/(\dbsize-\k)}$
        \STATE return $\widehat{F}, \widehat{\text{SSA}}, \widehat{\text{SSE}}$
    \end{algorithmic}
    \caption{private\_F($\x, \epsilon$)} 
     \label{alg:Fhat}
\end{algorithm}

Normally, the $F$ distribution is used to calculate a $p$-value for the $F$-statistic. However, Campbell et al.~find that the distribution of the private estimate $\widehat{F}$ differs too much from the $F$ distribution for this to be acceptable.  Furthermore, they find that it is no longer scale-free, meaning that the distribution depends on the within-group variance $\sigma^2$.  

Fortunately, the \sse is an estimate of $\sigma$, so using this estimate they computed an estimated distribution on $\widehat{F}$ through simulation. They could then compare a given value of $\widehat{F}$ to this distribution to obtain a $p$-value.

To assess the power of their private ANOVA algorithm, they simulate databases with three equal-sized groups with values drawn from $\normal(0.35, 0.15),$ $ \normal(0.5, 0.15)$, and $\normal(0.65, 0.15)$ respectively. For several $(\dbsize,\eps)$-pair choices, they generate many sets of data, apply the private ANOVA test, calculate the $p$-value, and record the percentage of simulations with $p$-values less than 0.05. They find that when $\eps = 1$, they need over five thousand data points to detect this effect (compared to two or three dozen data points in the public setting).  Our goal in this paper is to reduce the gap between the public and private setting.

\section{A New Test Statistic}\label{sec:f1}

In this section we describe our hypothesis test.  This begins with the introduction in Section \ref{subsec:f1def} of $F_1$, a new test statistic for the ANOVA setting.  In Section \ref{subsec:privf1} we then show how to privately compute a private approximation of $F_1$.  Finally, in Section \ref{subsec:alg-method} we calculate $p$-values for the private $F_1$ statistic.  This means simulating a correct reference distribution against which we can compare our output.

\subsection{The $F_1$ Statistic}\label{subsec:f1def}

Our goal is to define a statistic that releases similar information as the $F$-statistic, but has higher power (Definition~\ref{def:power}) for reasonable privacy guarantees.  We focused on two approaches to improve the power of the ANOVA calculation: reducing the amount of Laplacian noise by decreasing sensitivity, and making \ssa and \sse numerically larger so that the noise has less influence over the total value of the statistic.  We achieved both goals by taking the absolute values of the summand terms in the \ssa and \sse, rather than squaring them. 

As before, let \k denote the number of categories in the database, $c_i$ be the category and $y_i$ be the numerical value associated with observation $i$, $n_j$ be the size of category $j$, and $N$ be the total size of the data set. Additionally, $\bar{y}$ is the grand mean of the entries in database \x and $\bar{y}_j$ is the mean of the entries in group $j$. 
\begin{definition}[$F_1$-statistic]\label{def:f1stat}
Given a database \x with $k$ groups and \dbsize total entries, the $\sa(\x)$ and $\se(\x)$ calculations are defined as follows:
\begin{align*}
\sa(\x)  &= \sum_{j = 1}^{k} n_{j} \lvert \bar{y} - \bar{y}_{j} \rvert \\
\se(\x) &= \sum_{i=1}^N \lvert y_{i} - \bar{y}_{c_i} \rvert.
\end{align*}
The $F_1$-statistic is the ratio of \sa and \se, each divided by their respective degrees of freedom. 
\begin{equation*}
F_1(\x) = \frac{\sa(\x)/(k-1)}{\se(\x)/(\dbsize-k)}.
\end{equation*}
\end{definition}

The $F_1$-statistic measures variation between group means compared to variation within groups in essentially the same way as the original $F$-statistic. The \sa grows as the group means diverge.  What constitutes a ``large'' variation between group means depends on the variation between individual items, so \se, which measures this individual-level variation, provides a sense of scale for the \sa value.

In the next section, we show that the sensitivities of \sa and \se in the $F_1$-statistic are less than half as large as the sensitivities of \ssa and \sse in the original $F$-statistic. Further, because the summand terms are restricted to $[0,1]$, the \sa and \se values are larger, meaning they can tolerate the addition of more noise before losing their usefulness.

\subsection{A Private Approximation of $F_1$}\label{subsec:privf1}

The sensitivity of $F_1$ is very high.  (In the worst case, $\se(\x)$ is almost zero and very small changes can have huge effects on $F_1(\x)$.)  As a result, we can't simply apply the Laplace mechanism to $F_1$.  Instead, we choose to apply it individually to the \sa and \se functions, and then use composition and post-processing to compute an estimate of $F_1$.  We must therefore bound the sensitivities of \sa and \se.

We assume that the number of valid category values, \k, is fixed and public, but the number of entries in each group is not.  (This includes the possibility that one or more categories exist as valid entries but do not appear in the actual database.)   We also assume that there are maximum and minimum possible values for the data, and that the computation first uses these to normalize the data, mapping it to the interval $[0,1]$.

\begin{theorem}[SE Sensitivity] \label{thm:SEsens}
The sensitivity of the \se calculation in Definition~\ref{def:f1stat} is bounded above by 3.
\end{theorem}

\begin{proof}

Suppose neighboring databases \x and \xprime differ by some row $r$.
Say that in \x, $c_r = a$, and in \xprime, $c_r = b$. 
Let $n_a$ be the size of category $a$ excluding $r$ and let $n_b$ be the size of category $b$ excluding $r$. We begin by expressing the \se calculation as nested summations indexing over group size and entries within each group.

$$ \se(\x) = \sum_{j=1}^k \sum_{i \in C_j}  \lvert y_{i} - \bar{y}_{c_i} \rvert. $$

Call $t_{i} = \lvert y_{i} - \bar{y}_{c_i} \rvert$, and let $c_i \neq a,b$. Then, $t_{i}$ does not change between \x and \xprime. Now, suppose $i \ne r$ but $c_i = a$. It follows that $\Delta t_{i} \le 1/(n_a+1)$, since the only change comes from $\bar{y}_{a}$. There are $n_a$ such terms, so the total contribution from these terms is at most $n_a/(n_a+1)$. If $a \neq b$, we must also consider $t_{i}$ where $i \ne r$ but $c_i = b$, for which we have $\Delta t_{i} \le 1/(n_b+1)$. Thus, the terms in groups $a$ and $b$ excluding row $r$ together contribute 
$$ \frac{n_a}{n_a+1} + \frac{n_b}{n_b+1} \leq 2$$
if $a \neq b$ and just $n_a / (n_a+1) < 1$ otherwise.

Now, consider $\Delta t_{r}$. Since $y_{r}$, $\bar{y}_{a}$, and $\bar{y}_b$ are all in the interval $[0,1]$, the difference between $t_{r}$ in database \x and in \xprime is at most 1. Thus, the total sensitivity of \se is bounded above by 3.

\end{proof}

\begin{theorem}[SA Sensitivity] \label{thm:SAsens}
The sensitivity of the SA calculation in Definition~\ref{def:f1stat} is bounded above by 4. 
\end{theorem}

\begin{proof}
Again, suppose neighboring databases \x and \xprime differ by some row $r$ in both the categorical and numerical values, with $c_r = a$ in \x, and $c_r = b$ in \xprime. Denote the number of entries in groups $a$ and $b$ not including row $r$ by $n_a$ and $n_b$, respectively. We begin by expressing the \sa calculation as two sums indexing over groups and entries within each group.
$$
\sa(\x)  = \sum_{j = 1}^{k}\sum_{i \in C_j} \lvert \overline{y} - \overline{y}_{c_i} \rvert
$$

Consider the change from $\sa(\x)$ to $\sa(\xprime)$ as if it occurred in two steps.  In the first step, the grand mean $\overline{y}$ is updated.  Note that in the worst case, $\bar{y}$ can change by at most $1/\dbsize$ between \x and \xprime. Then, since there are $N$ summands, each including the grand mean, this step changes the value by a maximum of $N(1/N)=1$ to the overall sensitivity. 

In the second step we change the group means for groups $a$ and $b$.  There will be $n_a$ terms containing $\overline{y}_a$, each of which will change by at most $1/n_a$, changing the overall value by at most 1.  Similarly updating $\overline{y}_b$ changes $n_b$ terms each by at most $1/n_b$ for a total contribution of 1.  

Finally, the change of $r$ between \x and \xprime contributes 1 to the overall sensitivity, and thus the sensitivity of the \sa is bounded above by 4.

\end{proof}

Having proven these sensitivities, we can now introduce our private algorithm (Algorithm \ref{alg:F1}) for approximating $F_1$.  The algorithm first estimates \sa and \se, allocating part of the privacy budget to each one. We introduce a parameter $\rho \in (0,1)$ that determines the relative amount of the privacy budget spent on each intermediate value.  The optimal value of $\rho$ will be experimentally determined in Section \ref{subsec:optrho}.  We note that in addition to using a different test statistic, the prior work by Campbell et al.~did not consider $\rho$ values other than 0.5.

\begin{algorithm}
    \caption{private\_F1($\x,\eps, \rho$) \label{alg:F1}}
    \begin{algorithmic}
        \STATE $\widehat{\text{SA}} = \text{SA}(\x) + L_1$ where $L_1\sim\text{Lap}\left(\frac{4}{\rho\eps}\right)$ 
        \STATE $\widehat{\text{SE}} = \text{SE}(\x) + L_2$ where $L_2\sim\text{Lap}\left(\frac{3}{(1-\rho)\eps}\right)$
        \STATE  $\widehat{F_1} = \frac{\widehat{\text{SA}}/(k-1)}{\widehat{\text{SE}}/(N-k)}$
        \STATE return $\widehat{F_1}, \widehat{\text{SA}}, \widehat{\text{SE}}$
    \end{algorithmic}
\end{algorithm}

\begin{theorem} \label{thm:AlgPriv}
Algorithm \ref{alg:F1} is \eps-differentially private.
\end{theorem}
\begin{proof}
By the sensitivity bounds of \se and \sa in Theorems \ref{thm:SEsens} and \ref{thm:SAsens}  and the Laplace mechanism (Theorem \ref{thm:lapmechanism}), $\widehat{\sa}$ is $\rho\eps$-differentially private and $\widehat{\se}$ is $(1-\rho)\eps$-differentially private. By the composition theorem (Theorem \ref{thm:composition}), outputting both is \eps-differentially private. Since $k$ and \dbsize are both public information, computing $\widehat{F_1}$ is post-processing (Theorem \ref{thm:postprocessing}).
\end{proof}

\subsection{Reference Distribution and $p$-values}
\label{subsec:alg-method}

As discussed previously, the test statistic on its own is not useful; we need a $p$-value to provide a sense of scale in the context of the null hypothesis. Computing a $p$-value begins with an accurate reference distribution.  We numerically approximate this distribution through simulation.  The intermediate values of the $F$ statistic, \ssa and \sse, are drawn from  $\sigma^2\chi_{k-1}^2$ (for \ssa) and $\sigma^2\chi_{n-k}^2$ (for \sse).  Campbell et al.~\cite{campbell2018diffprivanova} used this to easily sample from the correct distributions for \ssa and \sse.

The \sa and \se values needed to compute $F_1$ follow no similarly tractable distribution, so instead we simulate full databases according to the null hypothesis and for each calculate $\widehat{F_1}$.  The distribution of these $\widehat{F_1}$ values approximates the reference distribution.  The goal, given a database \x, is to simulate databases with the same size $N$ and number of groups $k$, same standard deviation $\sigma$, and same expected value $\mu$.  The values $N$ and $k$ are public, so we can use those values.  The expected value $\mu$ is not, but as long as it is safely inside the $[0,1]$ interval, its value has no effect on the distribution, so we simply always use 0.5. 

Unfortunately, using an accurate $\sigma$ is more difficult.  Unlike in the non-private setting, the reference distribution depends on the choice of $\sigma$, so an inaccurate value can cause incorrect results.  We had two choices: either we could use some of our \eps-budget to directly estimate the standard deviation of the $y_i$ in database \x, or we could devise an indirect method of estimating the standard deviation given $\widehat{\sa}$ and $\widehat{\se}$. See Appendix~\ref{Sec:AppDirSig} for further discussion of the first option.  (It has higher power for low $n$ but takes longer to approach full power.)  Here we focus on the second option by deriving an unbiased estimator $\hat{\sigma}$ for $\sigma$ that can be computed from $\widehat{\se}$.

$$
\hat{\sigma} = \frac{\widehat{SE}}{\tilde{N}}\sqrt{\frac{\pi}{2}}, \text{ where } \tilde{N} = \sum_{j=1}^k n_j \sqrt{\left(1 - \frac{1}{n_j}\right)}
$$

See Appendix \ref{Sec:AppSig} for the proof that this is indeed an unbiased estimator.  

Computing this estimator requires knowledge of each of the group sizes, which are private. Fortunately, $\tilde{N}$ is closely approximated by $N-k$, which we know. At the smallest database sizes that we considered (around $N = 100$ with three equal-sized groups), this approximation has < 1\% error. As the size of the database grows, this error shrinks to zero.  Figure \ref{fig:sigma-estimate} visualizes the unbiasedness and asymptotically shrinking variance of $\hat{\sigma}$. (using the $\tilde{N} = N-k$ approximation).  We further confirm in Section 4 that it is precise enough to compute valid $p$-values.

\begin{figure}
\centering
\includegraphics[width=.7\linewidth]{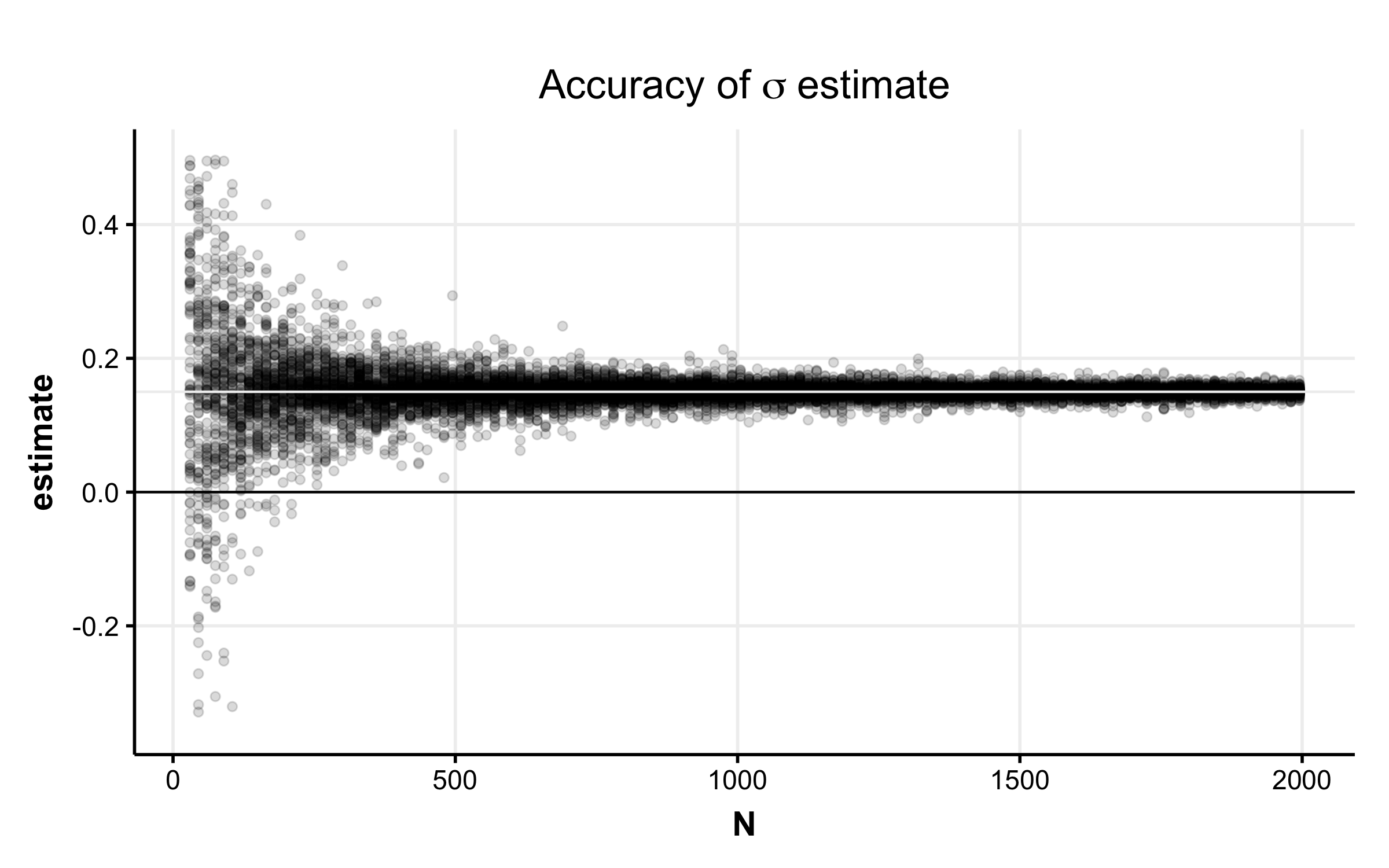}
\caption{$\hat{\sigma}$ is unbiased ($\sigma = .15$) with shrinking variance as N grows large. Each point represents the estimate from a simulated null database. At small N, there is a non-zero probability of returning a negative estimate.\label{fig:sigma-estimate}}
\end{figure}

Another issue presented by the private estimation of $\sigma$ is that the Laplacian noise can be large enough to make the estimate negative. Negative standard deviation estimates are more likely to occur when the \se is small, i.e.~when the database is small or when the within-group standard deviation is small.

This problem is unique to the private setting and we do the most conservative possible thing --- we never reject the null hypothesis when the estimated standard deviation was negative. Our reasoning was that a negative standard deviation has no statistical meaning and any calculations made from such an estimate would be uninformative.  This method for dealing with negative standard deviation estimates results in a type I error rate lower than $\alpha$, which means that our results are more conservative. As was mentioned earlier, the \sa increases as the database size and effect size between groups increases. Thus, negative standard deviation estimates tend to occur for database sizes and effect sizes that are so small that the effect is undetectable anyway.  As a result, this conservative choice does little to reduce the power of our test.

A formal description of the private $F_1$ ANOVA test in presented in Algorithm~\ref{alg:pval}, which returns a Boolean indicating whether the null hypothesis $H_0$ is rejected. 

\begin{algorithm}
    \caption{ANOVA\_test($\x$, $\eps$, $\alpha$, \emph{reps})\label{alg:pval} }
    \begin{algorithmic}
        \STATE $\widehat{F_1}, \widehat{\sa}, \widehat{\se} =$ private\_F1($\x$,$\eps$)
        \IF{$\widehat{\text{SE}}<0$}
        \STATE return \text{False}
        \ENDIF
        \STATE $\widehat{\sigma} = \sqrt{\pi/2}\frac{\widehat{\text{SE}}}{(\dbsize-k)}$
        \STATE $significant = 0$
        \FOR{$i = 1$ to \emph{reps}}
        \STATE  $\x^{ref} = N$ draws from $\normal(0.5, \widehat{\sigma})$ divided into $k$ equal-sized groups
        \STATE $\widehat{F_1}^{ref}, \widehat{\sa}^{ref}, \widehat{\se}^{ref} =$ private\_F1($\x^{ref}, \eps$)
        \IF{$\widehat{F_1}^{ref} > \widehat{F_1} $}
        \STATE $significant = significant + 1$
        \ENDIF
        \ENDFOR
        \STATE $p$-value $= significant/reps$
        \IF{$p$-value $< \alpha$}
        \STATE return \text{True}
        \ELSE 
        \STATE return \text{False}
        \ENDIF
    \end{algorithmic}
\end{algorithm}

\begin{theorem} \label{thm:fullTestPriv}
Algorithm \ref{alg:pval} is \eps-differentially private.
\end{theorem}
\begin{proof}
This follows immediately from Theorem \ref{thm:AlgPriv} (privacy of the test statistic) and Theorem \ref{thm:postprocessing} (post-processing).  Using the notation of Theorem \ref{thm:postprocessing}, private\_F1 (Algorithm \ref{alg:F1}) is the function $f$, which is proven $\epsilon$-differentially private by Theorem \ref{thm:AlgPriv}.  The rest of the computation is the function $g$, and uses only the output of $f$.  The whole of Algorithm \ref{alg:pval} is therefore the composition of these two functions and as a result is itself private.
\end{proof}

\section{Experimental Results}
\label{sec:results}

In this section we assess the properties and performance of Algorithm~\ref{alg:pval} through simulation. 

\subsection{Properties of the test}\label{subsec:properties}

A $p$-value is considered valid if ~\cite{casellaberger2002}

$$
\Pr[p\textrm{-value} \le \alpha \mid \X \leftarrow H_0] \le \alpha.
$$

In other words, the actual type I error rate must be less than or equal to $\alpha$. This can be assessed by conducting many tests on simulated data sampled according to $H_0$ at various $\alpha$ and checking if the proportion of rejections is less than or equal to $\alpha$. 

Figure~\ref{Fig:valid-pvals} presents these simulations for several choices of $\epsilon$ and demonstrates that the $F_1$ test produces valid $p$-values. The type I error rate is lowest at low $\epsilon$ values because the high privacy guarantee requires a large amount of noise be added to \se, which can lead to a negative estimate for $\sigma$, which in turns leads to an automatic decision to retain $H_0$ (see section~\ref{subsec:alg-method}).

\begin{figure}
\centering
\includegraphics[width=.8\linewidth]{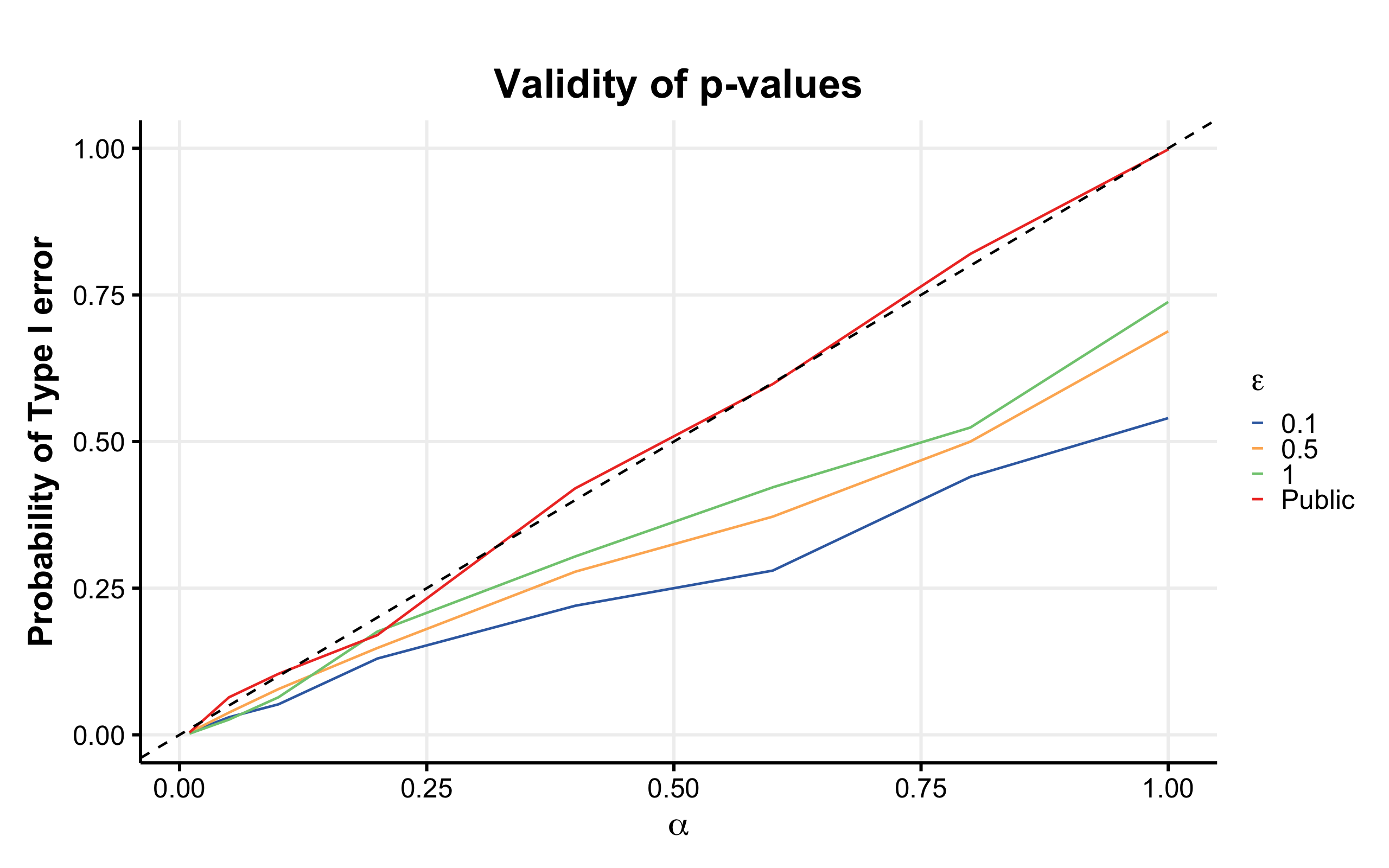}
\caption{The empirical type I error rate under three private scenarios is less than $\alpha$ while the public error rate is exactly $\alpha$ (within MC variability). Each point in a line represents 500 simulated tests, each with $N = 180$, $k=3$, and equal-size groups.\label{Fig:valid-pvals}}
\end{figure}

We note also that when we simulate data for calculating a reference distribution, we always simulate data with equal-size groups.  I.e., we must confirm that the critical value of the reference distribution is highest when groups are of equal size.  Fortunately, this appears to be the case.  Appendix \ref{app:unequal} contains both experimental and theoretical arguments for this claim, though not a complete analytic proof.

\subsection{Optimal \epsfrac}\label{subsec:optrho}

The computation of the private $F_1$ statistic requires the specification of $\epsfrac \in (0, 1)$, the parameter that determines the proportion of the privacy budget that is allocated to \sa relative to \se. We determined an optimal value for $\epsfrac$ by constructing power plots comparing database size to power for different \epsfrac values. We began by exploring the full range from $0.1$ to $0.9$ by $0.1$-increments to get a sense of the range of variability in power. After that initial pass, we tuned in to the value with higher precision. We considered many effect sizes and found that in every case $\epsfrac \approx 0.7$ was the most powerful.

Figure~\ref{Fig:f1-epsfrac} is an example of one of the many scenarios that were considered, which identifies .7 as the optimal value. The scale of the effect of $\rho$ on power is not dramatic but it was persistent across scenarios.

\begin{figure}
\centering
\includegraphics[width=.6\linewidth]{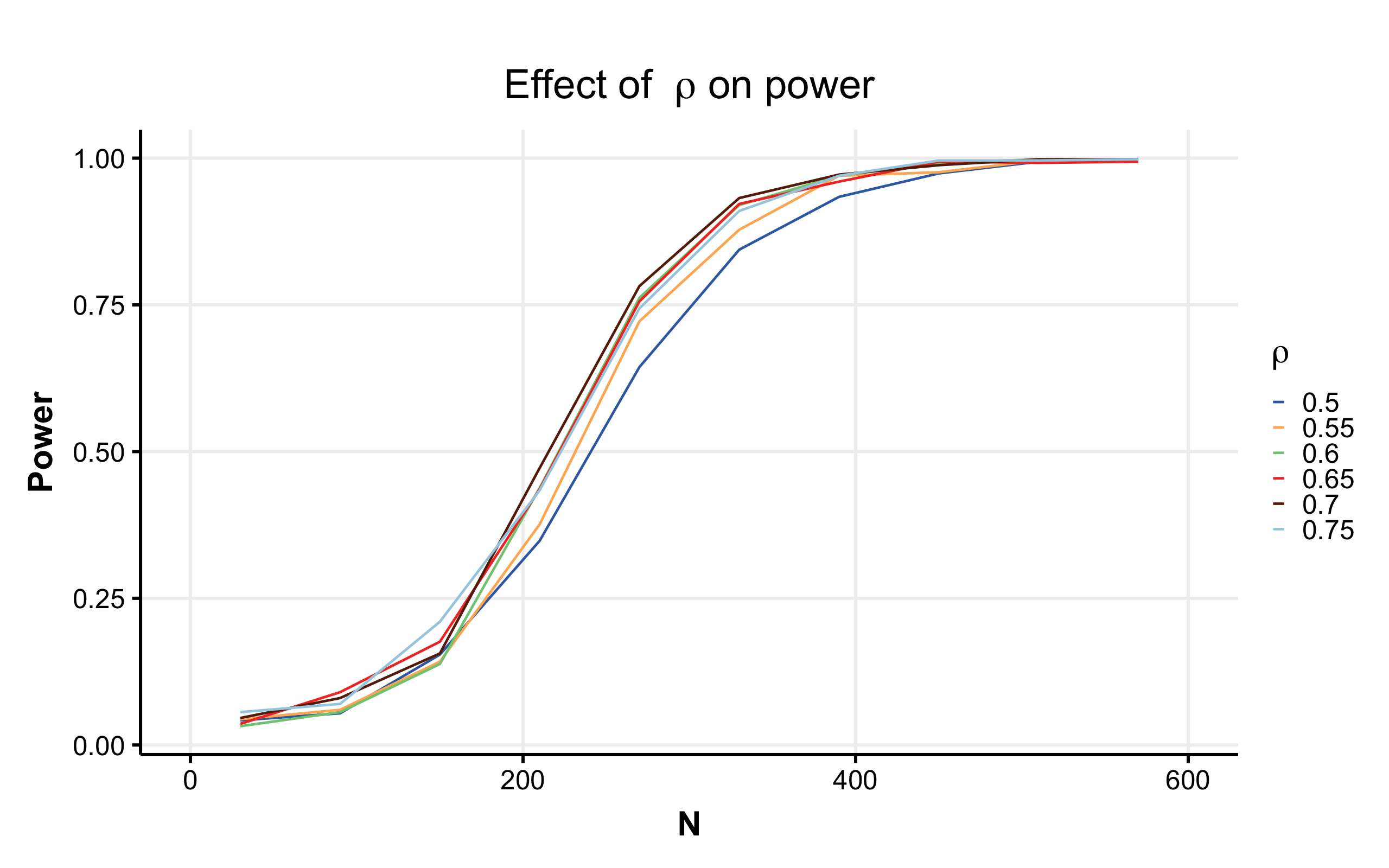}
\caption{Power curves at varying $\rho$ in a setting where $\epsilon = 1$, $k = 3$, $\sigma = .15$, and effect size: $1\sigma$. Power is experimentally maximized when $\rho = .7$.\label{Fig:f1-epsfrac}}
\end{figure}

\subsection{Performance of the test: power}
\label{subsec:power-analysis}

There can be many tests for a given scientific setting that generate valid $p$-values and have identical type I error rates. What distinguishes them is their statistical power, or the probability that they reject $H_0$ when the database comes from a distribution under $H_A$.

The most common way to assess the power of a test is to generate a plot of power as a function either of database size or of effect size. Figure~\ref{Fig:f1} fixes the effect size and shows power curves as a function of database size for four choices of $\epsilon$. An ideal test would very quickly develop power near 1 with very little data. In our private setting, it is clear that the cost of high privacy ($\epsilon = .1$) is roughly an order of magnitude more data than modest privacy ($\epsilon = 1$) to achieve high power. Our private test (for reasonable values of $\epsilon$) still requires much more data than the public version, hundreds of data points as opposed to dozens.

\begin{figure}
\centering
\includegraphics[width=.8\linewidth]{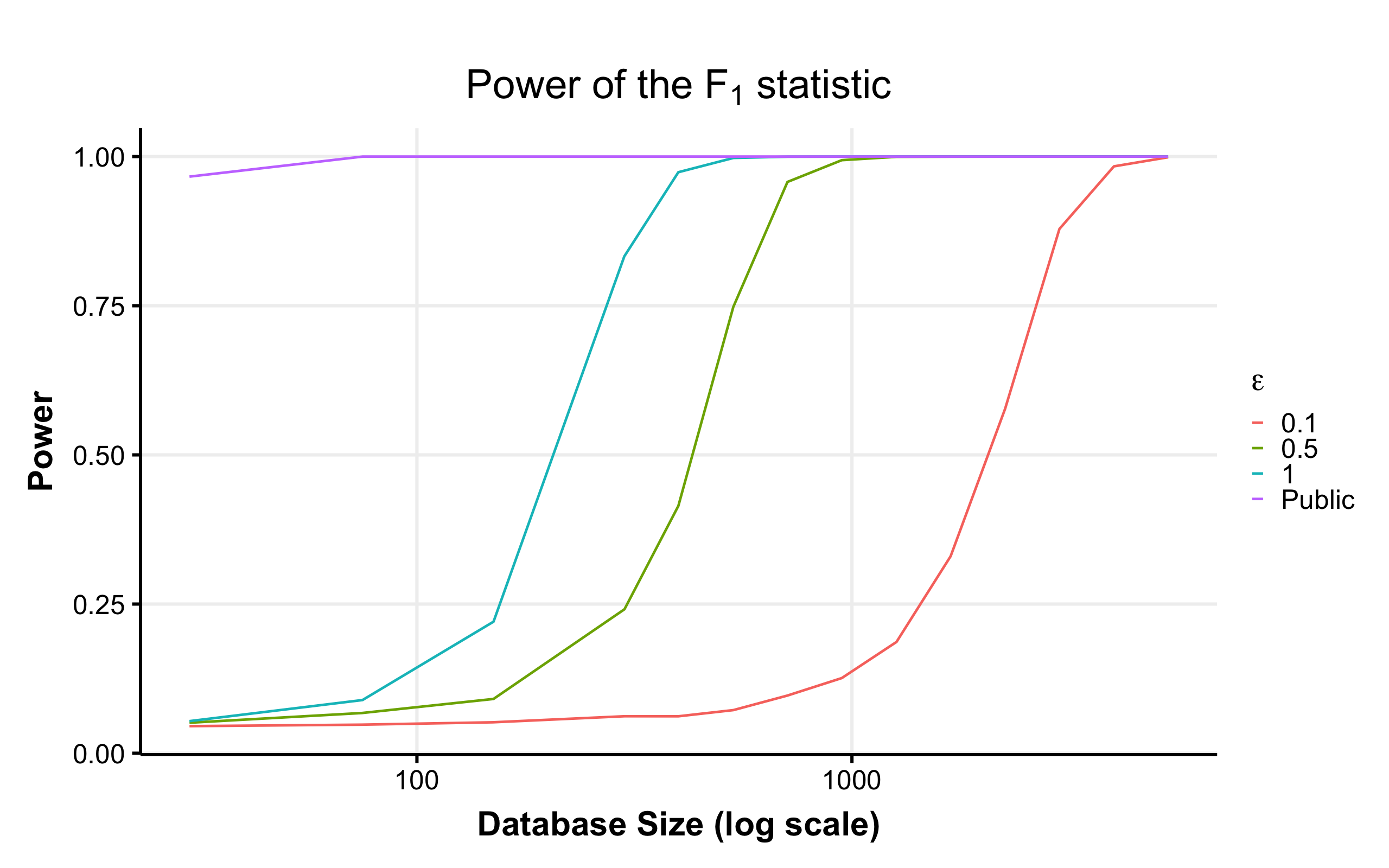}
\caption{Power curves at varying privacy levels in a simulation setting where $k = 3$, $\sigma = .15$, and effect size: $1\sigma$.\label{Fig:f1}}
\end{figure}

Each point in a line of a power curve was computed from 10,000 simulations, each based on a synthetic database of a set sample size and a set effect size (distance between the group means).  We considered the effect $H_A$ where $\k=3$ groups, each distributed $\normal(0.35, 0.15),  \normal(0.5,0.15), \normal(0.65, 0.15)$ and equal group sizes, which was the same scenario as Campbell et al.~\cite{campbell2018diffprivanova}.
Each simulation starts with a draw from $H_A$ and computes a $p$-value as described in Algorithm~\ref{alg:pval}.

The most relevant comparison, however, is between the private $F_1$ statistic and the privatized version of the classical $F_2$ statistic proposed in Campbell et al., the only prior private version of ANOVA.  This comparison can be seen in Figure \ref{Fig:f1-vs-f2}, and the improvement is substantial.  For example, at $\epsilon = 1$, if one wanted to collect enough data to detect this effect with 80\% probability, one would need ~4500 observations if using the prior best test, whereas with the $F_1$ test presented here one would need only ~300 data points, a 15-fold reduction in the necessary amount of data.

\begin{figure}
\centering
\includegraphics[width=.8\linewidth]{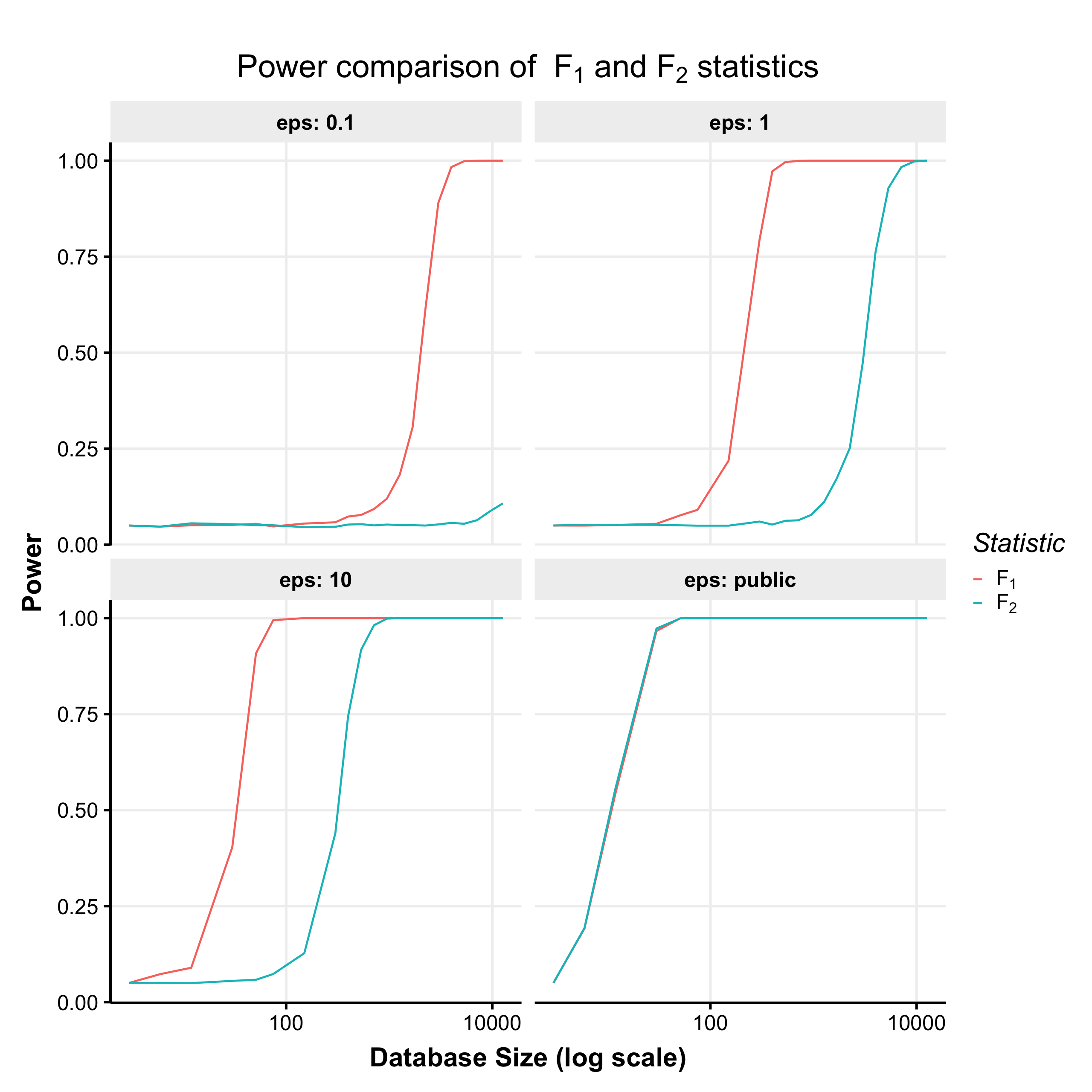}
\caption{Comparison of the power of the new $F_1$-statistic and the prior state of the art test using $F_2$, for three values of $\epsilon$ and in the public setting.\label{Fig:f1-vs-f2}}
\end{figure}

Figure \ref{Fig:f1-vs-f2} also demonstrates the degree to which the $F_1$ is well-suited to the private setting; the greatest improvement in power occurs under high privacy ($\epsilon = 0.1$). As $\epsilon$ grows large, the difference between the two statistics shrinks. In the public setting, they are nearly indistinguishable, though the $F_2$ is narrowly more powerful at every database size.\footnote{This empirical result is consistent with theoretical results in the classical statistics literature that discuss conditions in which the $F_2$ is a most powerful test~\cite{cox1974theoretical}.}

The improvement in power is the greatest practical contribution of our work: the ability to conduct a private ANOVA with an order of magnitude less data than the existing approach. This improvement can be attributed to two key characteristics of our new test. The first and most important is the notion of measuring distance using the $L^1$ norm.  The second is the unequal apportionment of the privacy between the \sa and \se terms.

\section{Other considerations}
\label{sec:considerations}

In this section we present a straightforward generalization of the classic $F$ and $F_1$ to allow other exponents. In Section \ref{subsec:fq} we show that in this generalized class of statistics $F_1$ is indeed optimal. In Section \ref{subsec:params} we discuss the experimental exploration of the full parameter space that guides the formulation of the $F_1$ statistic.

\subsection{Varying the exponent}\label{subsec:fq}
As seen is the previous sections, the change from squaring the differences in the original $F$-test to taking their absolute value with the $F_1$-statistic improved power significantly in the differentially private setting. It is not obvious that switching the exponent from 2 to 1 is optimal --- perhaps some other exponent is superior. An exponent of 0 is clearly horrible, so there must be a local maximum in the power of the statistic for some exponent between 0 and 2.

In order to determine which exponent is in fact optimal, we further generalize the notion of an $F$-test.  We define \sqa and \sqe, which are equivalent to \ssa and \sse except that the summand is raised to the $q^{\text{th}}$ exponent, and we call the resulting statistic $F_q$.  Note that $F_1$ (as defined earlier) is a special case of $F_q$ for $q=1$, and $F_2$ is the standard $F$-statistic.

\begin{definition}[$F_q$] \label{def:Fq} 
Given a database \x with \k groups and \dbsize total entries, define \sqa and \sqe as follows:
\begin{equation*}
\sqa(\x) = \sum_{j=1}^k n_j \left\vert \overline{y}_j - \overline{y} \right\vert^q
\end{equation*}
\begin{equation*}
\sqe(\x) = \sum_{i=1}^N \left\vert y_i - \overline{y}_{c_i} \right\vert^q
\end{equation*}
Then, $F_q$ is defined as
\begin{equation*}
F_q(\x) = \frac{\sqa/(\k-1)}{\sqe/(\dbsize-\k)}
\end{equation*}
\end{definition}
We must now create a private approximation of $F_q$ for arbitrary $q$.  To do this, we first bound the sensitivity of the \sqa and \sqe with the following two theorems, the proofs of which can be found in Appendix \ref{sec:fqsensitivity}.
\begin{theorem}[\sqe Sensitivity] \label{thm:SQEsens} 
The sensitivity of \sqe is bounded above by
\begin{equation*}
2\bigg(\frac{\dbsize}{2}\bigg)^{(1-q)} + 1
\end{equation*}
when $q \in (0,1)$ and
\begin{equation*}
\dbsize - \dbsize\bigg(1-\frac{2}{\dbsize}\bigg)^q +1 
\end{equation*}
when $q\geq 1$. Note that both give an upper bound of 3 when $q=1$.
\end{theorem}

\begin{theorem}[\sqa Sensitivity]\label{thm:SQAsens} The sensitivity of \sqa is bounded above by 
\begin{equation*}
\dbsize\bigg(\frac{3}{\dbsize}\bigg)^q + 1
\end{equation*}
when $q \in (0,1)$ and
\begin{equation*}
\dbsize-\dbsize\bigg(1-\frac{3}{\dbsize}\bigg)^q + 1
\end{equation*}
when $q \geq 1$. Note that both give an upper bound of 4 when $q = 1$.
\end{theorem}

Given these sensitivity bounds, we can calculate a private approximation of $F_q$ for simulated data using the same algorithm as for $F_1$, but with the sensitivities altered according to the choice of $q$.  We can also simulate a reference distribution by adapting Algorithm~\ref{alg:pval}, which was done to construct Figure \ref{fig:fqpower}.  As is clear from these results, in terms of power, the optimal value of $q$ is 1.

We note that in the computation shown in Figure \ref{fig:fqpower}, we don't estimate $\sigma$ using \sqe.  This is because we have not developed an estimator for $\sigma$ that can be computed from \sqe (for $q \ne 1, 2$).  If another value of $q$ had indeed been optimal, the next step would have been to find such an estimator and confirm it was accurate enough to produce acceptable $p$-values.  But since the power cannot possibly improve when switching to an estimated $\sigma$ value, this result is sufficient to show that other values of $q$ need not be considered.

\begin{figure}
\centering
\includegraphics[width=\linewidth]{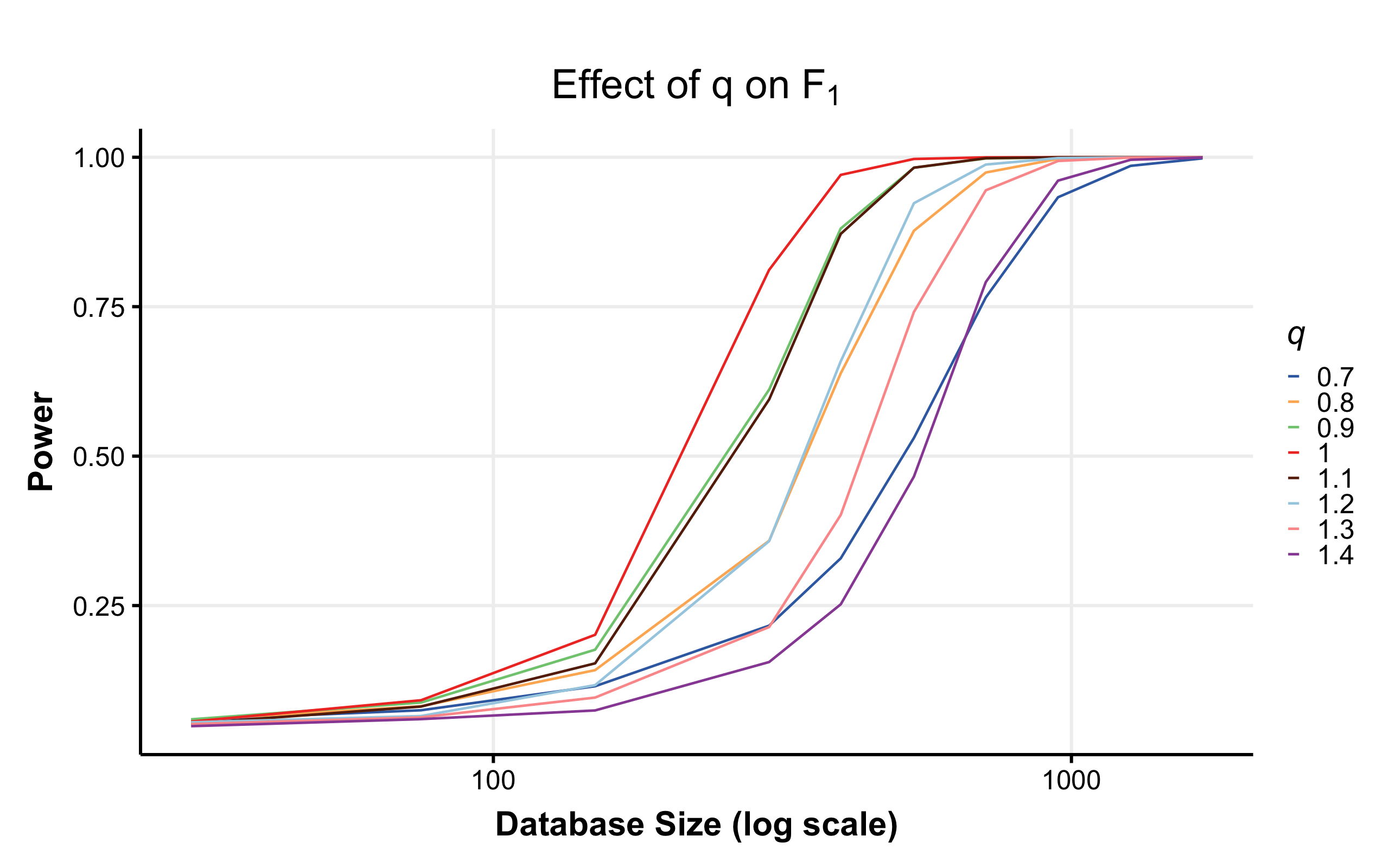}
\caption{Power curves at varying exponents in a simulation setting where $k = 3$, $\sigma = .15$, and effect size: $1\sigma$. Power is experimentally maximized when $q = 1$.}\label{fig:fqpower}
\end{figure}

\subsection{Parameter tuning}\label{subsec:params}

With the generalization of the $F$-statistic, we add $q$ to the list of parameters that determine the power of a testing procedure. The parameters can be organized as follows:

\vspace{3mm}
\textbf{Data Generation:} $N$, $k$, $\sigma$, effect size \\
\indent\textbf{Private Algorithm:} $\epsilon$, $q$, $\rho$
\vspace{3mm}

The analyst gets to select the parameters corresponding to the private algorithm. While
$\epsilon$ is set based on privacy concerns, $q$ and $\rho$ should be set to 
maximize power, which our work suggest occurs at roughly 1 and 0.7, respectively.
This conclusions is based upon an extensive exploration of the parameter space, 
a selection of which can be found in Appendix~\ref{Sec:AppOptPars}.

The salient feature of these plots is that the choice of $q$ is much more 
consequential than the choice of $\rho$. In the setting where $\epsilon = .1$, 
we found that the result seen in Figure~\ref{Fig:f1-vs-f2} -- a greater than 
10-fold reduction in database size to get equivalent power -- holds across a range
of difference data generation parameters.

By contrast, the effect of $\rho$ on power is much smaller; the database size
reduction is closer to 1.1- or 1.2-fold when moving from $\rho = .5$ to $\rho = .7$
when $\epsilon = .1$.

\section*{Acknowledgements}
This material is based upon work supported by the National Science Foundation under Grant No. SaTC-1817245 and by the Gillespie Family Student Research Fund.

\bibliographystyle{plain}
\bibliography{ANOVA}

\appendix

\section{Appendix: An unbiased estimator for $\sigma$}
\label{Sec:AppSig}

We are concerned with finding an unbiased estimator for $\sigma$ given $\widehat{SE}$.

\begin{align}
\widehat{SE} = \left(\sum_{j=1}^k \sum_{i \in C_j} \lvert y_i - \bar{y}_{c_i} \rvert \right) + L_2 \label{eq:se}
\end{align}

Our first goal is to find $E(\widehat{SE})$. We begin by noting several relevant distributions.
\begin{align*}
y_i &\sim N(\mu, \sigma^2) \\
\bar{y}_{c_i} &\sim N(\mu, \sigma^2/n_{c_i}) \\
y_i - \bar{y}_{c_i} &\sim N(0, \tau^2).
\end{align*}

The term $\tau^2$ will be useful in determing the expected value of each term in $\widehat{SE}$, so we seek to express it in terms of known quantities.
\begin{align}
\tau^2 &= Var(y_i - \bar{y}_{c_i}) \nonumber \\
&= Var(y_i) + Var(\bar{y}_{c_i}) - 2 Cov(y_i, \bar{y}_{c_i}) \nonumber \\
&= \sigma^2 + \frac{\sigma^2}{n_{c_i}} - 2 Cov(y_i, \bar{y}_{c_i}). \label{eq:tau}
\end{align}

To solve for the covariance, noting that $Cov(X,Y) = E(XY) - \mu^2$, we start by finding
\begin{align}
E(y_i \bar{y}_{c_i}) &= \frac{1}{n_{c_i}} E(y_i (y_1 + y_2 + \ldots + y_i + \ldots + y_{n_{c_1}})) \nonumber \\
&= \frac{1}{n_{c_i}} E(y_i^2 + y_i S_{-i}) \nonumber \\
&= \frac{1}{n_{c_i}}\left( E\left(y_i^2\right) + E\left(y_i S_{-i}\right)\right), \label{eq:eyy}
\end{align}

where $S_{-i} = y_1 + y_2 + \ldots + y_{i-1} + y_{i+1} + \ldots + y_{n_{c_i}}$. The term $y_i^2$ will be chi-square distributed when standardized as follows:
\begin{align*}
Z^2 = \left(\frac{y_i - \mu}{\sigma}\right)^2 = \frac{1}{\sigma^2}\left(y^2 - 2y_i\mu + \mu^2\right).
\end{align*}

\noindent Therefore we can write
\begin{align*}
E(y^2) &= E\left(\frac{\sigma^2}{\sigma^2} \left(y^2 - 2y_i\mu + \mu^2\right) + 2y_i\mu - \mu^2\right) \\
&= E\left(\sigma^2 Z^2 + 2y_i\mu - \mu^2\right) \\
&= \sigma^2E(Z^2) + 2\mu E(y_i) - \mu^2 \\
&= \sigma^2 + 2\mu^2 - \mu^2 \\
&= \sigma^2 + \mu^2.
\end{align*}

\noindent Continuing from Eq.~\eqref{eq:eyy}, 
\begin{align*}
\frac{1}{n_{c_i}}\left( E\left(y_i^2\right) + E\left(y_i S_{-i}\right)\right) &= \frac{1}{n_{c_i}}\left(\sigma^2 + \mu^2 + E(y_i)E(S_{-i})\right) \\
&= \frac{1}{n_{c_i}}\left(\sigma^2 + \mu^2 + (n_{c_i} - 1)\mu^2\right) \\
&= \frac{1}{n_{c_i}}\left(\sigma^2 + n_{c_i}\mu^2\right) \\
&= \frac{\sigma^2}{n_{c_i}} + \mu^2.
\end{align*}

\noindent Now we can return to the covariance term in Eq.~\eqref{eq:tau},
\begin{align*}
Cov(y_i, \bar{y}_{c_i}) &= E(y_i \bar{y}_{c_i}) - \mu^2 \\
&= \frac{\sigma^2}{n_{c_i}} + \mu^2 - \mu^2 \\
&= \frac{\sigma^2}{n_{c_i}}.
\end{align*}

\noindent Now we can finish the calculation of $\tau^2$.
\begin{align*}
\tau^2 &= \sigma^2 + \frac{\sigma^2}{n_{c_i}} - 2 Cov(y_i, \bar{y}_{c_i}) \\
&= \sigma^2 + \frac{\sigma^2}{n_{c_i}} - 2 \frac{\sigma^2}{n_{c_i}} \\
&= \sigma^2 - \frac{\sigma^2}{n_{c_i}} \\
&= \sigma^2\left(1 - \frac{1}{n_{c_i}}\right).
\end{align*}

The distribution of the absolute value of a normal random variable with mean 0 and variance $\tau^2$ is half normal with a single parameter $\tau$, which should properly be indexed by the observation. For ease of notation, call this random variable $W_i$.
\begin{align*}
W_i = \lvert y_i - y_{c_i}\rvert \sim HN(\tau_i),
\end{align*}

\noindent where $E(W_i) = \tau_i \sqrt{\frac{2}{\pi}}$.  To find $E(\widehat{SE})$, we can write it using the double sum notation as in Eq.~\eqref{eq:se}. Since $L_2$ is a Laplace distribution centered at zero, we can use linearity of expectation to further simplify the expectation:
\begin{align*}
E(\widehat{SE}) &= E\left(\sum_{j=1}^k \sum_{i \in C_i} W_i\right)+ E(L_2)\\
&= E\left(\sum_{j=1}^k \sum_{i \in C_i} W_i \right) + 
0\\
&= \sum_{j=1}^k E\left(n_j W_j \right).
\end{align*}

\noindent The change in indices is justified by realizing that $W_i$ is the same for all $n_{c_i}$ elements in $C_i$ (in expectation). Continuing,
\begin{align*}
\sum_{j=1}^k E\left(n_j W_j \right) &= \sum_{j=1}^k n_j E\left(W_j \right) \\
&= \sum_{j=1}^k n_j \tau_j \sqrt{\frac{2}{\pi}} \\
&= \sqrt{\frac{2}{\pi}} \sum_{j=1}^k n_j \sqrt{\sigma^2\left(1 - \frac{1}{n_j}\right)}  \\
&= \sigma \sqrt{\frac{2}{\pi}} \underbrace{\sum_{j=1}^k n_j \sqrt{\left(1 - \frac{1}{n_j}\right)}}_{\tilde{N}}.
\end{align*}

\noindent Denote the sum $\tilde{N}$, which is a number that approaches $N$ as the group sizes get large. This allows us to express the expected value more concisely as
\begin{align*}
E(\widehat{SE}) &= \sigma \sqrt{\frac{2}{\pi}} \tilde{N}.
\end{align*}

The final step is to correct for this bias in our final estimator:
\begin{align*}
\hat{\sigma} &\coloneqq \frac{\widehat{SE}}{\tilde{N}}\sqrt{\frac{\pi}{2}}, \text{where}\\
E\left(\hat{\sigma}\right) &= \sigma.
\end{align*}

\noindent Using this exact estimator requires knowledge of each of the group sizes, which are private. Instead of dedicating part of the $\epsilon$ budget to this estimation, we used $N - k$ in place of $\tilde{N}$. At the smallest database sizes that we considered (around $N = 100$), this approximation accounts for $< 1\%$ error. As the size of the database grows, this error shrinks to zero.

\section{Appendix: Validity under unequal $n_j$}\label{app:unequal}

Consider a specific setting in which $N = 800$, $\sigma = 0.1$, and $k = 4$. Figure~\ref{Fig:unequal-group-sizes} shows the reference distributions of $F_1$ in four scenarios, each with a different group allocation of the 800 observations. The distribution in $s_0$, the equal group size scenario, generally takes the highest values; indeed it exceeds the other scenarios at every quantile. This represents the distribution that we use to calculate $p$-values and reject $H_0$ whenever an observed statistic is greater than the vertical dotted line (when $\alpha = .05$). This means that in the other three scenarios, featuring unequal allocation, the actual type I error rate (the proportion of the distributions beyond the dotted line) will be less than $\alpha$ and we meet the condition for valid $p$-values.

\begin{figure}
\centering
\includegraphics[width=.85\linewidth]{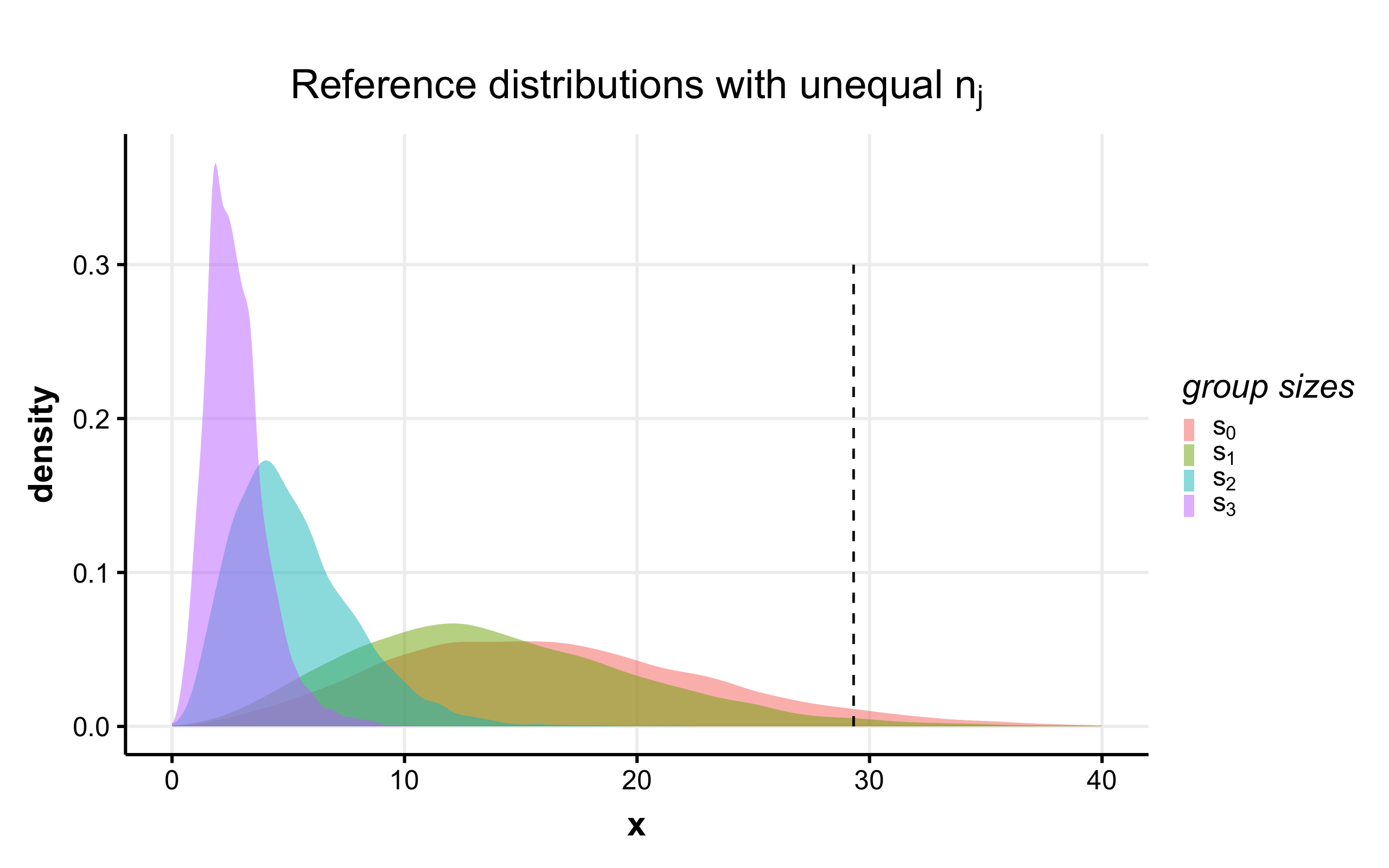}
\caption{Distribution of $\hat{F}_1$ under $H_0$ in four group size allocation scenarios, each with 10,000 simulations. $s_0\colon \{200, 200, 200, 200 \}$, $s_1\colon \{100, 100, 100, 500 \}$, $s_2\colon \{5, 10, 20, 765 \}$, $s_3\colon \{3, 3, 3, 791 \}$. The vertical dotted line indicates the .95 quantile of the equal group size distribution.\label{Fig:unequal-group-sizes}}
\end{figure}

In the following, we show that under any scenario, the expected value of the $F_1$ statistic is maximized when the group sizes are equal. Recall that the form of the statistic is
$$
F_1(\x) = \frac{\sa(\x)/(k-1)}{\se(\x)/(\dbsize-k)}.
$$

where
\begin{equation*}
\sa(\x) = \sum_{j=1}^k n_j |\yj - \yb|
\end{equation*}

and
\begin{equation*}
\se(\x) = \sum_{i=1}^N \lvert y_{i} - \bar{y}_{c_i} \rvert.
\end{equation*}

Only $SA$ is a function of the $n_j$, so we restrict our attention to that term as $E(F_1)$ for different allocations will scale with $E(SA)$ by a multiplicative constant.

Denote each term in the summation $SA_j$. By using the approach used in Appendix~\ref{Sec:AppSig} to find the distribution of $y_i - y_{c_i}$, it can be shown that $\yj - \yb$ is distributed normal mean zero, and variance $\omega_j^2 = \frac{\sigma^2}{n_j} - \frac{\sigma^2}{N}$.
\begin{align*}
    \E(SA) &= \E\left(\sum_{j=1}^k SA_j\right) \\
    &= \sum_{j=1}^k \E(SA_j) \\
    &= \sum_{j=1}^k\E\left(n_j|\yj - \yb|\right) \\
    &= \sum_{j=1}^k\E\left(n_j|\frac{\omega_j(\yj - \yb)}{\omega_j}|\right) \\
    &= \sum_{j=1}^kn_j\omega_j\E\left(|\frac{(\yj - \yb)}{\omega_j}|\right) &&\text{since } \omega_j > 0 \\
    &= \sum_{j=1}^kn_j\omega_j(\sqrt{\frac{2}{\pi}})
\end{align*}

Since each $|\frac{(\yj - \yb)}{\omega_j}|$ is distributed standard half normal, the expectation of each evaluates to the constant $\sqrt{\frac{2}{\pi}}$. Now we may also factor out $\sigma$ from $\omega_j$, providing the final answer:
\begin{equation*}
    \E(SA) = \sigma(\sqrt{\frac{2}{\pi}})\sum_{j=1}^kn_j\sqrt{\frac{1}{n_j} - \frac{1}{N}}.
\end{equation*}

\begin{prop}
The expectation of $SA$ is maximized when group sizes are all equal, i.e., when $n_j = n_{j'}$, for all $j \neq j'$
\end{prop}

\begin{proof}
It suffices to show that the sum between any two summands of the expectation is maximized when the other $k-2$ group sizes are fixed, so we begin by fixing each $n_j$ except for $n_1$ and $n_2$, without loss of generality. Define $c = n_1 + n_2$ to be the sum between the two unfixed group sizes, so then the sum between the first two summands becomes:
\begin{align*}
&\sigma(\sqrt{\frac{2}{\pi}})(n_1\sqrt{\frac{1}{n_1} - \frac{1}{N}} + (c-n_1)\sqrt{\frac{1}{(c-n_1)} - \frac{1}{N}}) \\
&=\sigma(\sqrt{\frac{2}{\pi}})(\sqrt{n_1 - \frac{n_1^2}{N}} + \sqrt{(c-n_1) - \frac{(c-n_1)^2}{N}})
\end{align*}
since $n_1$ and $c-n_1 > 0$. Now factoring constants and differentiating with respect to $n_1$:
\begin{align*}
\frac{1}{\sigma(\sqrt{\frac{2}{\pi}})}\frac{d}{dn_1}\E(SA_1 \!+\! SA_2) =~~~~~~&\\
 \frac{\frac{2 (c-n_1)}{N}-1}{2 \sqrt{-\frac{(c-n_1)^2}{N}+c-n_1}}&+\frac{1-\frac{2 n_1}{N}}{2 \sqrt{n_1-\frac{n_1^2}{N}}},
\end{align*}
which has only one solution in $n_1$, when $n_1 = \frac{c}{2} =  \frac{n_1 + n_2}{2}$, which implies $n_1 = n_2$. This critical point is associated with a maximum expectation on the first two summands, as desired. Applying this result across all $n_j, n_{j'}$ pairs such that $j \neq j'$, will result in maximizing each summand of the expectation when all $n_j = n_{j'}$, hence maximizing $\E(SA)$.
\end{proof}

\section{Appendix: Sensitivity Proofs for \sqa and \sqe}\label{sec:fqsensitivity}

Recall that, from Definition \ref{def:Fq}, we have the following functions \sqa and \sqe:
\begin{equation*}
\sqa(\x) = \sum_{j=1}^k n_j \left\vert \overline{y}_j - \overline{y} \right\vert^q
\end{equation*}

\begin{equation*}
\sqe(\x) = \sum_{i=1}^N \left\vert y_i - \overline{y}_{c_i} \right\vert^q
\end{equation*}

\noindent Our goal here is to bound the sensitivities for \sqa and \sqe.

\setcounter{theorem}{6}
\begin{theorem}[\sqe Sensitivity]  \label{thm:SQEsens-appendix} 
The sensitivity of \sqe is bounded above by
\begin{equation*}
2\bigg(\frac{\dbsize}{2}\bigg)^{(1-q)} + 1
\end{equation*}
when $q \in (0,1)$ and
\begin{equation*}
\dbsize - \dbsize\bigg(1-\frac{2}{\dbsize}\bigg)^q +1 
\end{equation*}
when $q\geq 1$. Note that both give an upper bound of 3 when $q=1$.
\end{theorem}
\begin{proof}
As in the previous proofs of the sensitivity of \se and \sa, suppose neighboring databases \x and \xprime differ by some row $r$, with $c_r = a$ in \x and $c_r = b$ in \xprime. Rewrite the \sqe as a sum that indexes over group size and entries within each group.
$$\sqe(\x) = \sum_{j=1}^k \sum_{i \in C_j} \left\vert  y_i - \overline{y}_{c_i} \right\vert^q.$$
Let $t_i =  \left\vert y_{i} - \overline{y}_{c_i} \right\vert ^q$ for any entry $i$.  Note that if $c_i \neq a,b$, then $t_i$ will not change between databases \x and \xprime, as the group means of the other groups are not altered. Thus, unless $i$ is in group $a$ or $b$, $t_i$ will contribute nothing to the overall sensitivity of the \sqe. For notational ease, let $z = y_i-\overline{y}_{c_i}$ and parameterize $t_i$ as a function of $z$. The maximum $z$ can change by is $1/n_{c_i}$. We now bound the sensitivity by individually bounding the sensitivity $\Delta t_i$ of each term with  $c_i = a,b$.   

\smallskip\noindent\textbf{Case 1:}
When $q$ is less than $1$, $t_i(z)$ is a concave function with minimum at $z=0$ that is symmetric about the y-axis and which monotonically increases for positive $z$. Because the slope of $t_i(z)$ is highest near $z=0$, the worst case sensitivity is between $z = 0$ and $z = 1/n_{c_i}$, and hence
\begin{align*}
\Delta t_{i} &\le \left\vert t_i(0) - t_i(1/n_{c_i}) \right\vert \\
	&= (1/n_{c_i})^q.
\end{align*}

Note that when $c_i = a$ and $i \neq r$, $\Delta t_i \leq (1/n_a)^q$.  The analogous statement holds for group $b$.  By multiplying these bounds by the number of terms in each group, we get
\begin{equation*}
\Delta \sqe \le  n_a^{1-q} + n_b^{1-q} + 1,
\end{equation*}
where the first and second terms are the total change possible to terms in groups $a$ and $b$ respectively, and the final term is for the contribution of row $r$ itself, which we cannot bound other than by noting that its value both in \x and \xprime falls inside of $[0,1]$.

We must now take the worst-case value of this bound over all possible database sizes $n_a$ and $n_b$.  Since $n_a$ is a positive integer and $q<1$, $n_a^{1-q}$ increases for a given $q$ as $n_a$ increases (and the same is true for $n_b$). So, the worst-case sensitivity will occur when as much of the total database is in groups $a$ and $b$ as possible. Write $n_b = N-n_a$. Then
$$n_a^{1-q} + (N-n_a)^{1-q}$$
is a downward facing parabola-like function with maximum value when $n_a = N/2$. So, the sensitivity of the \sqe is bounded above by

$$\Delta \sqe \le 2\left(\frac{N}{2}\right)^{1-q} + 1.$$

\noindent\textbf{Case 2:}
When $q$ is greater than $1$, $t_{i}$ is a convex function with maximum at $z=1$, symmetric about the $y$-axis, and monotonically increasing for positive $z$. So, the worst case sensitivity is between $z=1$ and $z=1-1/n_{c_i}$, and hence 
\begin{align*}
\Delta t_{i} &\le \left\vert t_{i}(1) - t_{i}(1-1/n_{c_i}) \right\vert \\
	&= 1 - (1-1/n_{c_i})^q.
\end{align*}
Summing these bounds over all terms, we have 
$$ \Delta\sqe \le n_a(1-(1-1/n_a)^q) + n_b(1-(1-1/n_b)^q) + 1.$$

Again, the sensitivity will be maximized when as much of the database is distributed between $n_a$ and $n_b$ as possible. To determine what the worst case allocation is, let 
$$f = n_a(1-(1-1/n_a)^q) + (N-n_a)(1-(1-1/(N-n_a))^q) + 1,$$
i.e., $f$ is an expression for the upper bound of $\Delta\sqe$ with $n_b$ replaced by $N-n_a$ to maximize sensitivity. Then, we can maximize $f$ in terms of $n_a$:
$$ \frac{\partial \Delta f}{\partial n_a} =  -\frac{1-q}{(N - n_a)^q} + \frac{1-q}{n_a^q}$$
has a critical point at $n_a = N/2$, and 
$$ \frac{\partial \Delta^2 f}{\partial n_a^2} = - \frac{(1-q)q}{(N-n_a)^{-1-q}} - \frac{(1-q)q}{n_a^{-1-q}}$$
is always negative. So, $f$ is concave down and $n_a = N/2$ is a global maximum. Hence, the worst case sensitivity occurs when the database is distributed equally between groups $a$ and $b$, i.e.,

$$ \Delta\sqe \le N \left( 1- \left(1-\frac{2}{N}\right)^q \right) + 1.$$
\end{proof}

\setcounter{theorem}{7}
\begin{theorem}[\sqa Sensitivity] \label{thm:SQAsens-appendix} The sensitivity of \sqa is bounded above by 
\begin{equation*}
\dbsize\bigg(\frac{3}{\dbsize}\bigg)^q + 1
\end{equation*}
when $q \in (0,1)$ and
\begin{equation*}
\dbsize-\dbsize\bigg(1-\frac{3}{\dbsize}\bigg)^q + 1
\end{equation*}
when $q \geq 1$. Note that both give an upper bound of 4 when $q = 1$.
\end{theorem}

\begin{proof}
Let $s_{j} =  \left\vert \bar{y}_{j} - \grand \right\vert ^q$ for any group $j$. I.e., $s_{j}$ is the (unweighted) term in the calculation of the \sqa that corresponds to group $j$. Note that as the grand mean changes between databases $\x$ and $\xprime$ in addition to the group means, all terms, not just those for groups $a$ and $b$, will contribute to the sensitivity of the \sqa. Recall that the sensitivity of the grand mean is $1/N$, while the sensitivity of the group mean for groups $a$ and $b$ are $1/n_a$ and $1/n_b$ respectively. \\

\noindent\textbf{Case 1:} When $q<1$, $s_j(z)$ is a concave function with minimum at $z=0$, symmetric about the $y$-axis, and monotonically increasing for positive $x$. So, the worst case sensitivity of $\Delta s_j$ for $j \ne a,b$ is between $x=0$ and $x=1/N$, and the worst case sensitivity of $\Delta s_j$ for $j = a,b$ is between $z=0$ and $z=1/n_j + 1/N$. Then, the total sensitivity of the \sqa for $q<1$ is
\begin{align*}
 \Delta\sqa \le  (N-n_a-&n_b-1)(1/N)^q + n_a(1/N + 1/n_a)^q \\
 &+ n_b(1/N + 1/n_b)^q + 1 .
\end{align*}

The addition of the $1$ comes from the fact that our data point $r$ that switches between groups contributes $\vert \bar{y}_a - \grand \vert^q$ to the calculation of the \sqa in database \x, and contributes $\vert \bar{y}_b - \grand \vert^q$ to the calculation of the \sqa in database \xprime; the difference between these two terms is bounded above by 1. Note that since $q<1$, $(1/z)^q > 1/z$. Hence, $(1/N + 1/n_a)^q > (1/N)^q$ and thus the worst-case sensitivity occurs when all of $N$ is allocated to groups $a$ and $b$. I.e.,
$$ \Delta\sqa \le n_a(1/N + 1/n_a)^q + (N-n_a)(1/N + 1/(N-n_a))^q. $$
Then, as in the proof of the \sqe's sensitivity, to determine the worst-case sensitivity in terms of $N$, let 
$$g = n_a(1/N + 1/n_a)^q + (N-n_a)(1/N + 1/(N-n_a))^q $$
and maximize this expression in terms of $n_a$.
\begin{align*}
\frac{\partial g}{\partial n_a} = -&\left(\frac{1}{N} + \frac{1}{N-n_a}\right)^q + \left(\frac{1}{N} + \frac{1}{n_a}\right)^q \\
&+ \frac{q\left(\frac{1}{N} + \frac{1}{N-n_a}\right)^{q-1}}{N-n_a} - \frac{q\left(\frac{1}{N} + \frac{1}{n_a}\right)^{q-1}}{n_a}
\end{align*}
This is a symmetric expression between $N$ and $N-n_a$. So, there must be a critical point at $n_a = N/2$. Note also that
\begin{align*}
\frac{\partial^2 g}{\partial n_a^2} &= N^2(q-1)q \left( \frac{(\frac{1}{N} + \frac{1}{N-n_a})^q}{(N-n_a)(-2N+n_a)^2} + \frac{\left(\frac{1}{N} + \frac{1}{n_a}\right)^q}{n_a(N+n_a)^2}\right) \\
	&\le 0,
\end{align*}
since $N>n_a$ and $q<1$. So, $n_a = N/2$ is a global maximum, and hence
$$\Delta\sqa \le N \left( \frac{3}{N} \right)^q.$$

\noindent\textbf{Case 2:} When $q$ is greater than $1$, $s_j$ is a convex function with minimum at $z=0$, symmetric about $z=0$, and monotonically increasing for positive $z$. So, the worst case sensitivity of $\Delta s_j$ for $j \ne a,b$ is between $z= 1$ and $z=1-1/N$, and the worst-case sensitivity of $\Delta s_j$ for $i=a,b$ is between $z=1$ and $z = 1 - 1/N - 1/n_i$. Then, the total sensitivity of the \spa for $q<1$ is
\begin{align*}
\Delta\sqa &\le  (N-n_a-n_b)(1-(1-1/N)^q) \\
    & \hspace{1cm} + n_a(1-(1-1/N-1/n_a)^q) \\
	& \hspace{1cm} + n_b(1-(1-1/N-1/n_b)^q) + 1.
\end{align*}

Note that 
\begin{align*}
1-1/N-1/n_b &< 1-1/N \\
\Rightarrow (1-1/N-1/n_b)^q &< (1-1/N)^q \\
\Rightarrow 1 - (1-1/N-1/n_b)^q &> 1- (1-1/N)^q.
\end{align*}
Thus, $\Delta\sqa$ is maximized when $N$ is maximally allocated to groups $a$ and $b$. As in the proof for $q<1$, this occurs when $n_a = N/2 = n_b$. Then, 
$$\Delta\sqa < N \left( 1 - \left( 1 - \frac{3}{N}\right)^q \right).$$
\end{proof}

\section{Appendix: Direct Calculation of $\sigma$}
\label{Sec:AppDirSig}

In our work, we used $\widehat{SE}$ to form an estimator for $\sigma$ when calculating the null distribution.  We developed another version of a differentially-private ANOVA that calculates $\sigma$ directly using a portion of the privacy budget.  We first define a generalized variance and prove the sensitivity bounds on this quantity.

\begin{definition}[$\var_q$] \label{def:varq} Given a database \x with $k$ groups and $n_j$ entries in the $j$-th group, the $\var_q$ calculation is defined as
\begin{equation*}
\var_q = \sum_{i=1}^N \lvert y_{i} - \overline{y} \rvert^q
\end{equation*}
where $q$ is a positive real number.
\end{definition}

\begin{theorem}[\varq-Sensitivity] \label{thm:varqSens} 
The sensitivity of \varq is bounded above by
$$ \frac{N-1}{N^q} + 1 $$
when $q<1$, and is bounded above by
$$ (N-1)(1-(1-1/N)^q) + 1 $$
when $q>1$. Note that these both give a bound of $2-1/N$ when $q=1$.
\end{theorem}

\begin{proof}
Consider databases \x and \xprime which differ in entry $r$. Recall that the sensitivity of the grand mean is $1/N$. Let $t_i = \left\vert  y_i - \grand \right\vert ^q$. When $q<1, t_i$ is a concave function with positive range, so the worst case sensitivity is between $ y_i - \grand = 0$ and $ y_i - \grand = 1/N$. That is, 
\begin{align*}
\Delta t_i &= \left\vert t_{i}(0) - t_i\left(\frac{1}{N}\right) \right\vert \\
	&= \left( \frac{1}{N} \right)^q.
\end{align*}

Every single term can be affected by at most $\left( \frac{1}{N} \right)^q$, except the term $r$, which can change by $1$. So, 
\begin{align*}
\Delta \varq &\le \left\vert (N-1)(\frac{1}{N})^q + 1 \right\vert \\
	&= \frac{N-1}{N^q} + 1.
\end{align*}
When $q>1, t_i$ is convex. The worst case sensitivity is between  $y_i - \grand = 1$ and $y_i - \grand = 1-1/N$ Then,
\begin{align*}
\Delta \varq \le \left\vert (N-1)(1-(1-\frac{1}{N})^q) + 1 \right\vert.
\end{align*}
\end{proof}

We again use the Laplace mechanism, and the sensitivity of a database's variance follows directly from Thm.~\ref{thm:varqSens}:

\begin{corollary}
\label{thm:varsens}
The sensitivity of the variance of a database is bounded above by
$$ 3 + 1/N^2 - 3/N.$$
\end{corollary}

Algorithm~\ref{alg:F1hatVar} divides the privacy budget into $\rho_1$ for the \sa, $\rho_2$ for the \se, and $\rho_3$ for the $\var$ calculations respectively.  The values of $\rho$ are provided as input along with the database \x and the $\eps$ value. When $\widehat{SE}$ is used as the $\sigma$ estimate, we found a 70-30 split of the privacy budget between \sa and \se was optimal (Figure~\ref{Fig:f1-epsfrac}).   In the power analysis of Algorithm~\ref{alg:F1hatVar}, we vary the proportion of the privacy budget dedicated to the \var calculation.  We fixed the proportion of the privacy budget for $\rho_1$ and $\rho_2$ to be a 70-30 split budget not used by $\rho_3$ (the \var calcuation).  

\begin{algorithm}
    \begin{algorithmic}
        \STATE \textbf{Input:} Database \x, $\eps$ value, $\rho_1, \rho_2, \rho_3$
        \STATE Compute $\widehat{\sa} = \sa + Z_1$ where $Z_1\sim \lap\left(\frac{4}{\eps \rho_1}\right)$
        \STATE Compute $\widehat{\se} = \se+ Z_2$ where $Z_2\sim \lap\left(\frac{3}{\eps \rho_2}\right)$
        \STATE Compute $\!\!\widehat{\var} \!\!= \!\!\var \!\!+\! Z_3\!$ where  $\!Z_3 \!\sim\! \lap\!\!\left(\! \frac{3\!+\!1/N^2\! - \!3/N}{\eps\rho_3} \!\right)$
        \STATE Compute $\widehat{F_1} = \frac{\widehat{\sa}/(\k-1)}{\widehat{\se}/(\dbsize-\k)}$
        \STATE \textbf{Output:} $\widehat{F_1}, \widehat{\sa}, \widehat{\se}, \widehat{\var}$
    \end{algorithmic}
    \caption{Differentially private $F_1$-statistic with direct calculation of variance}
     \label{alg:F1hatVar}
\end{algorithm}

Unsurprisingly, allocating part of the budget to calculating the standard deviation negatively impacts the statistical power of the test (Fig.~\ref{Fig:multiple-rhos}). However, it is surprising that this allocation does not impact the power of the test more strongly, since \var has a larger sensitivity and requires adding noise to smaller values than the \sa and \se. 

\begin{figure}[h]
\centering
\includegraphics[width=\linewidth]{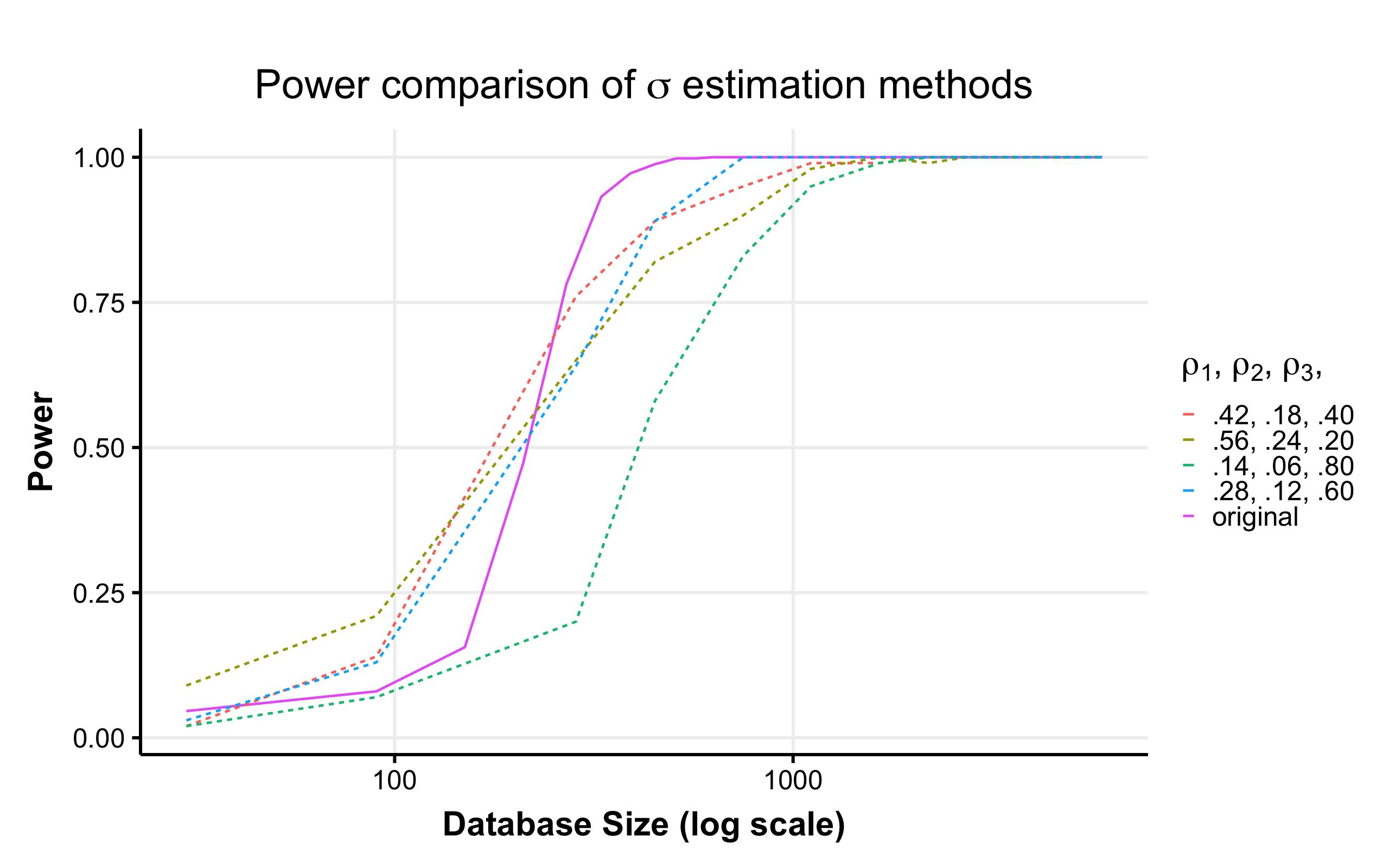}
\caption{The power of $F_1$ for a fixed $\epsilon$ value budgeted across \sa, \se, and \var calculations, with the portion of $\epsilon$ not allocated to $\rho_3$ allocated between $\rho_1$ and $\rho_2$ with a 70-30 split.\label{Fig:multiple-rhos}}
\end{figure}

\section{Appendix: Parameter Selection}
\label{Sec:AppOptPars}

Empirical results guiding parameter selection are presented on the following
pages in Figure~\ref{fig:par-plot-1} ($\epsilon$ = .1) and Figure~\ref{fig:par-plot-2} 
($\epsilon$ = 1). 

\begin{figure*}[h]
\centering
\includegraphics[width=\linewidth]{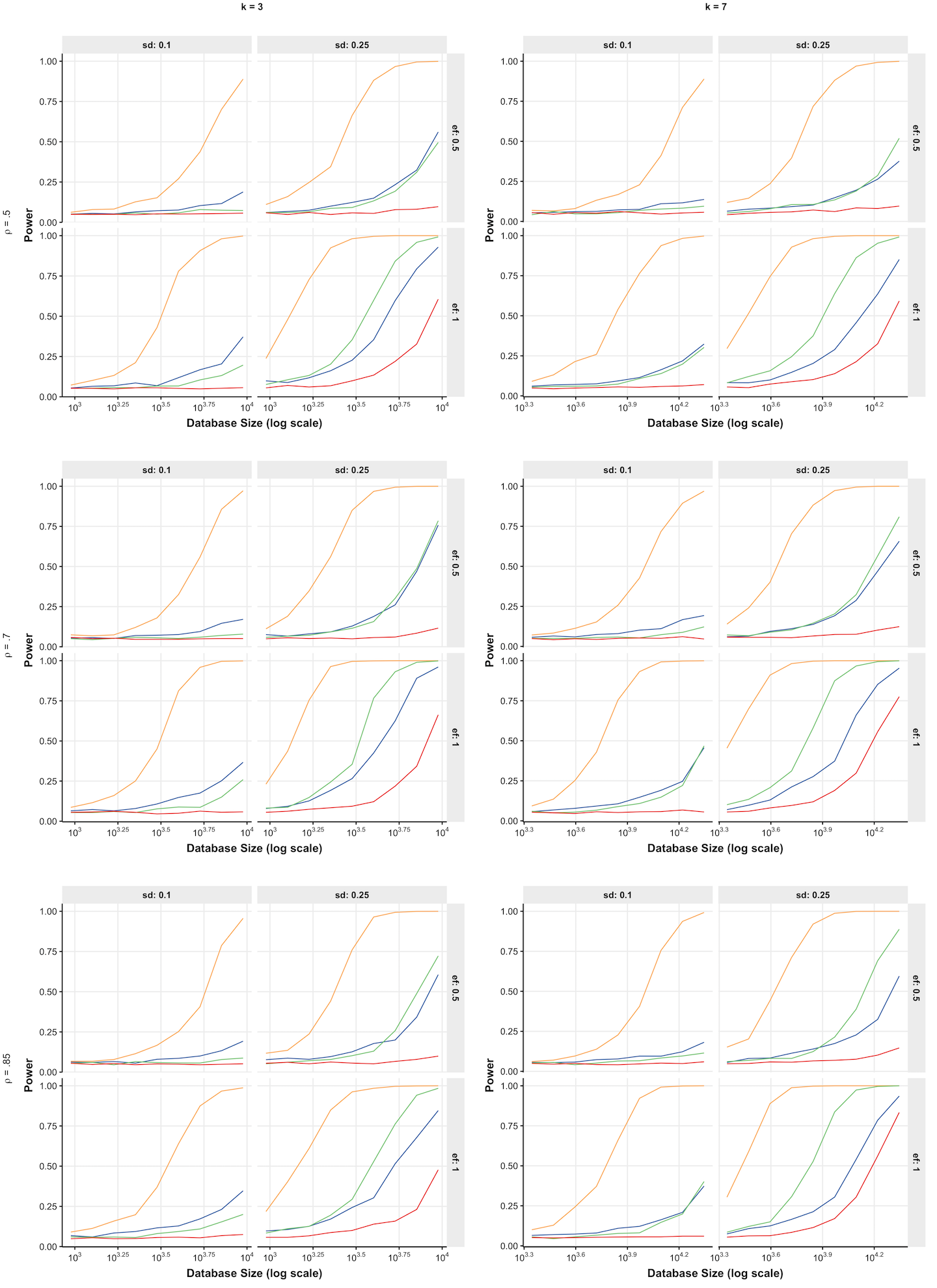}
\caption{Power comparison in settings where $\epsilon = .1$. Within each subplot,
each color curve corresponds to a different value of $q$ (blue: 0.75, gold: 1,
green: 1.5, red: 2). The two main columns
of plots correspond to different values of $k$ and the three main rows correspond
to difference values of $\rho$. Note that the scale on the x-axis differs with $k$
($k = 7$ requires more data).}\label{fig:par-plot-1}
\end{figure*}

\begin{figure*}[h]
\centering
\includegraphics[width=\linewidth]{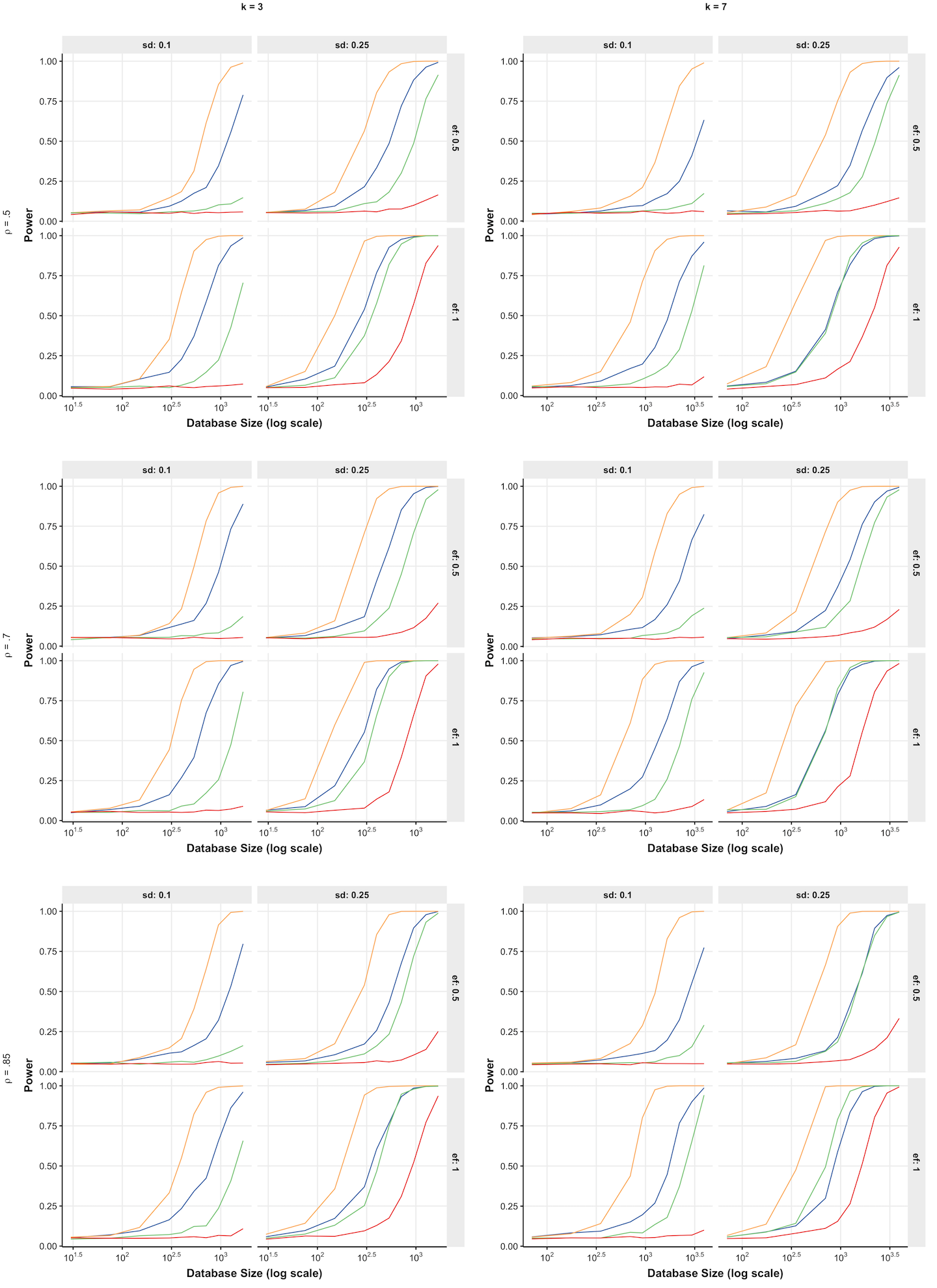}
\caption{Power comparison in settings where $\epsilon = 1$. Within each subplot,
each color curve corresponds to a different value of $q$ (blue: 0.75, gold: 1,
green: 1.5, red: 2). The two main columns
of plots correspond to different values of $k$ and the three main rows correspond
to difference values of $\rho$. Note that the scale on the x-axis differs with $k$
($k = 7$ requires more data).}\label{fig:par-plot-2}
\end{figure*}

\end{document}